\documentclass[12pt,a4paper,oneside]{amsart}

\usepackage{amssymb,amsmath}
\usepackage{graphicx,psfrag}
\usepackage{stmaryrd}
\usepackage{color}

\begin{document}

\newcommand{\red}[1]{\textcolor{red}{#1}}
\newcommand{\?}{{\bf ???}}

\definecolor{thmcolor}{rgb}{0,0,.4} 
\definecolor{remarkcolor}{rgb}{0,.2,0} 
\definecolor{proofcolor}{rgb}{.4,0,0} 
\definecolor{quecolor}{rgb}{.2,.2,0} 
\definecolor{axcolor}{rgb}{.3,0,.3}
\definecolor{thmbgcolor}{rgb}{0.9,0.9,1} 
\definecolor{rmbgcolor}{rgb}{0.9,1,0.9} 
\definecolor{proofbgcolor}{rgb}{1,0.9,0.9} 
 
\renewcommand{\proofname}{\colorbox{proofbgcolor}{\textcolor{proofcolor}{Proof}}} 
\renewcommand{\qedsymbol}{\colorbox{proofbgcolor}{\textcolor{proofcolor}{$\blacksquare$}}} 
 
\theoremstyle{definition} \newtheorem{thm}{\colorbox{thmbgcolor}{\textcolor{thmcolor}{Theorem}}}[section] 
\theoremstyle{definition} \newtheorem{cor}[thm]{\colorbox{thmbgcolor}{\textcolor{thmcolor}{Corollary}}} 
\theoremstyle{definition} \newtheorem{lem}[thm]{\colorbox{thmbgcolor}{\textcolor{thmcolor}{Lemma}}}
\theoremstyle{definition} \newtheorem{prop}[thm]{\colorbox{thmbgcolor}{\textcolor{thmcolor}{Proposition}}}
 
\theoremstyle{remark} \newtheorem{conv}[thm]{\colorbox{rmbgcolor}{\sc\textcolor{remarkcolor}{Convention}}} 
\theoremstyle{remark} \newtheorem{conj}[thm]{\colorbox{rmbgcolor}{\sc\textcolor{remarkcolor}{Conjecture}}}
\theoremstyle{definition} \newtheorem{que}[thm]{\colorbox{rmbgcolor}{\textcolor{quecolor}{Question}}} 
\theoremstyle{definition} \newtheorem{rem}[thm]{\colorbox{rmbgcolor}{\textcolor{remarkcolor}{Remark}}}
\theoremstyle{definition} \newtheorem{example}[thm]{\colorbox{rmbgcolor}{\textcolor{remarkcolor}{Example}}}

\newcommand{\limit}{\mathsf{Limit}}
\newcommand{\Hyp}{\mathsf{Hyp}}
\newcommand{\taxis}{\mathsf{t\text{-}axis}}
\newcommand{\com}{\mathsf{com}}
\newcommand{\sh}{\mathsf{sh}}
\newcommand{\ch}{\mathsf{ch}}
\newcommand{\sqslope}{\mathsf{slope}^2}
\newcommand{\sqdist}{\mathsf{dist}^2}
\newcommand{\timed}{\mathsf{time}}
\newcommand{\sqspace}{\mathsf{space}^2}
\newcommand{\sqlength}{\mathsf{length}^2}
\newcommand{\sqspeed}{\mathsf{speed}^2}
\newcommand{\ls}{\mathsf{c}}
\newcommand{\q}{\mathsf{q}}
\newcommand{\vo}{\bar o}
\newcommand{\vp}{\bar p}
\newcommand{\vq}{\bar q}
\newcommand{\vr}{\bar r}
\newcommand{\vx}{\bar x}
\newcommand{\vy}{\bar y}
\newcommand{\vz}{\bar z}
\newcommand{\vv}{\bar v}
\newcommand{\vu}{\bar u}
\newcommand{\vw}{\bar w}
\newcommand{\ve}{\bar e}
\newcommand{\va}{\bar a}
\newcommand{\vex}{\bar 1_x}
\newcommand{\vet}{\bar 1_t}
\newcommand{\then}{\rightarrow}
\newcommand{\oszt}{\slash}
\newcommand{\gyok}{\sqrt{\phantom{n}}}
\newcommand{\de}{:=}
\newcommand{\rac}{\mathbb{Q}}
\newcommand{\R}{\mathbb{R}}
\newcommand{\A}{\mathbb{A}}

\newcommand{\defiff}{\ \stackrel{\;def}{\Longleftrightarrow}\ }

\newcommand{\IOb}{\ensuremath{\mathsf{IOb}}} 
\newcommand{\IB}{\ensuremath{\mathsf{IB}}} 
\newcommand{\Ob}{\ensuremath{\mathsf{Ob}}} 
\newcommand{\B}{\ensuremath{\mathit{B}}} 
\newcommand{\Ph}{\ensuremath{\mathsf{Ph}}} 
\newcommand{\F}{\ensuremath{\mathsf{F}}} 
\newcommand{\Q}{\ensuremath{\mathit{Q}}} 
\newcommand{\W}{\ensuremath{\mathsf{W}}} 
\newcommand{\E}{\ensuremath{\mathsf{E}}} 
\newcommand{\SpecRel}{\ensuremath{\mathsf{SpecRel}}} 
\newcommand{\AccRel}{\ensuremath{\mathsf{AccRel}}} 
\newcommand{\GenRel}{\ensuremath{\mathsf{GenRel}}} 
\newcommand{\ev}{\ensuremath{\mathsf{ev}}} 
\newcommand{\cl}{\ensuremath{\mathsf{cl}}} 
\newcommand{\Apr}{\ensuremath{\mathsf{Apr}}} 
\newcommand{\Id}{\ensuremath{\mathsf{Id}}} 
\newcommand{\dom}{\ensuremath{\mathsf{Dom}\,}} 
\newcommand{\ran}{\ensuremath{\mathsf{Ran}\,}} 
\newcommand{\M}{\ensuremath{\mathfrak{M}}} 
\newcommand{\ax}[1]{\textcolor{axcolor}{\ensuremath{\mathsf{#1}}}} 
\newcommand{\AxEv}{\ensuremath{\mathsf{AxEv}}}
\newcommand{\wl}{\ensuremath{\mathsf{wl}}}
\newcommand{\lc}{\ensuremath{\mathsf{lc}}}
\newcommand{\w}{\ensuremath{\mathsf{w}}}
\newcommand{\num}{\textit{Num}}

\title{What are the numbers in which spacetime?}
\author{H.~Andr\'eka, J.~X.Madar\'asz, I.~N\'emeti, G.~Sz\'ekely}
\date{\today}

\begin{abstract}
Within an axiomatic framework, we investigate the possible structures
of numbers (as physical quantities) in different theories of
relativity.
\end{abstract}

\maketitle

\section{Introduction}

Basically, we would like to investigate the following metaphysical question:

\centerline{\it What are the numbers in the physical world?}
\smallskip

Without making this question more precise we can make the following
two natural guesses which contradict each other:
\begin{itemize}
\item Obviously, the physical numbers are  the real (or the complex) numbers since at least 99\% of our physical theories are using these numbers.
\item Obviously, the set of physical numbers is a subset of the rational numbers (or even the integers) since the outcomes of the measurements have finite decimal representations. 
\end{itemize}

Clearly, this informal level is too naive to meaningfully investigate
our question. However, that does not mean that it is impossible to
scientifically investigate our question within some logical
framework. In this paper, we are going to reformulate and investigate
this question (restricted to spacetime theories) within a rigorous
logical framework.

First of all, what do numbers have to do with the geometry of
spacetime?  The concepts related to numbers can be defined by the
concepts of geometry by Hilbert's coordinatization, see, e.g.,
\cite[pp.23-27]{Gol}. Moreover, purely geometrical statements can
correspond to statements about the structure of numbers. For example,
in Cartesian planes over ordered fields, the statement {\it ``every line
which contains a point from the interior of a circle intersects the
circle''} is equivalent to that {\it``every positive number has a square
root,''} see, e.g., \cite[Prop.16.2., p.144]{hartshorne}.
In the spirit of this example, here we investigate the question 
\begin{center}
{\it``How are some 
properties of spacetime reflected\\ on the structure
  of numbers?''}
\end{center}

Among others, we will see axioms on observers also implying that
positive numbers have square roots.  Ordered fields in which positive
numbers have square roots are called {\bf Euclidean fields}, which got
their names after their role in Tarski's first-order logic
axiomatization of Euclidean geometry \cite{TarskiElge}.

Let \ax{Th} be a theory of space-time which
contains the concept of numbers (as physical quantities) together with
some algebraic operations on them, such as addition ($+$),
multiplication ($\cdot$) (or at least these concepts are definable in
\ax{Th}.). In this case, we can introduce notation $\num(\ax{Th})$
for the class of the quantity parts (quantity structures) of the models of theory \ax{Th}:
\begin{equation*}
\num(\ax{Th})=\{\text{The quantity parts of the models of \ax{Th}}\}.
\end{equation*}
We use the notation $\mathfrak{Q}\in\num(\ax{Th})$ for algebraic
structure $\mathfrak{Q}$ the same way as the model theoretic notation
$\mathfrak{Q}\in Mod(\ax{AxField})$, e.g., $\rac\in\num(\ax{Th})$
means that $\rac$, the field of rational numbers, can be the
structure of quantities (numbers) in \ax{Th}.
Now we can scientifically investigate the question   
\begin{center}
\it ``What are the numbers in physical theory \ax{Th}?''
\end{center}
by studying what algebraic structures occur in $\num(\ax{Th})$.  

In this paper, we investigate this question only in the case when
\ax{Th} is a theory of spacetimes. However, this question can be
investigated in any other physical theory the same way.

We will see that the answer to our question often depends on the
dimension of spacetime. Therefore, we will introduce notation
$\num_{n}(\ax{Th})$ at page \pageref{num} for the class of the
possible quantity structures of theory \ax{Th} if all the investigated
spacetimes are $n$-dimensional.

In the logic language of Section~\ref{lang-s}, we will introduce
several theories and axioms of relativity theory.  For example, our
starting axiom system for special relativity (called \ax{SpecRel},
see page \pageref{specrel}) captures the kinematics of
special relativity perfectly, see Theorem~\ref{thm-poi} and Corollary
\ref{cor-slow}.  Furthermore, without any extra assumptions
\ax{SpecRel} has a model over every ordered field, i.e.,
\begin{equation*}
\num(\ax{SpecRel})=\{\,\text{ordered fields} \,\},
\end{equation*}
see Remark~\ref{rem-of}. Therefore, \ax{SpecRel} has a model over
$\rac$, too.  However, if we assume that inertial observes can move
with arbitrary speed less than that of light (in any direction every
where), see \ax{AxThExp} at page \pageref{axthexp}, then every
positive number has to have a square root if $n\ge 3$ by
Theorem~\ref{thm-eof}, i.e.,
\begin{equation*}
\num_n(\ax{SpecRel} + \ax{AxThExp})=\{\text{\,Euclidean  fields\,}\}.
\end{equation*}
In particular, the number structure cannot be the field of rational
numbers, but it can be the field of real algebraic numbers.

We will also see that our axiom system of special relativity has a
model over $\rac$ if we assume axiom \ax{AxThExp} only approximately
(which is reasonable as we cannot be sure in anything perfectly
accurately in physics), see Theorem~\ref{thm-rac},
Corollary~\ref{cor-arch} and Conjecture~\ref{conj-of}.

It is interesting that, if the spacetime dimension is 3, then we do
not need the symmetry axiom of \ax{SpecRel} to prove that every
positive number has a square root if \ax{AxThExp} is assumed, see
Theorem~\ref{thm-eof3}. However, in even dimensions, it is
possible that some numbers do not have square roots, see
Theorem~\ref{thm-eof4} and Questions~\ref{que-thx} and \ref{que-2}.

Moving toward general relativity we will see that our theory of
accelerated observes (\ax{AccRel}) requires the structure of
quantities to be a real closed field, i.e., a Euclidean field in
which every odd degree polynomial has a root, see
Theorem~\ref{thm-rc}. However, any real closed field, e.g., the field
of real algebraic numbers, can be the quantity structure of \ax{AccRel}. 

If we extend \ax{AccRel} by extra axiom \ax{Ax\exists UnifOb} stating
that there are uniformly accelerated observers, then the field of real
algebraic numbers cannot be the structure of quantities any more if
$n\ge3$, see Theorem~\ref{thm-noalg}.  A surprising consequence of this
result is that $\num_n(\ax{AccRel}+\ax{Ax\exists UnifOb})$ is not a
first-order logic definable class of fields, see Remark~\ref{rem-noalg}. 

In Section~\ref{sec-gr}, we introduce an axiom system of general
relativity \ax{GenRel} and investigate our question a bit for
\ax{GenRel}.

\section{The language of our theories}
\label{lang-s}

To investigate our reformulated question, we need an
axiomatic theory of spacetimes.  The first important decision in
writing up an axiom system is to choose the set of basic
symbols of our logic language, i.e., what objects and relations between
them we will use as basic concepts.

Here we will use the following two-sorted\footnote{That our theory is
  two-sorted means only that there are two types of basic objects
  (bodies and quantities) as opposed to, e.g., Zermelo--Fraenkel set
  theory where there is only one type of basic objects (sets).}
language of first-order logic (FOL) parametrized by a natural number
$d\ge 2$ representing the dimension of spacetime:
\begin{equation*}
\{\, \B,\Q\,; \Ob, \IOb, \Ph,+,\cdot,\le,\W\,\},
\end{equation*}
where $\B$ (bodies\footnote{By bodies we mean anything which can move,
  e.g., test-particles, reference frames, electromagnetic waves,
  centers of mass, etc.}) and $\Q$ (quantities) are the two sorts,
$\Ob$ (observers), $\IOb$ (inertial observers) and $\Ph$ (light
signals) are one-place relation symbols of
sort $\B$, $+$ and $\cdot$ are two-place function symbols of sort
$\Q$, $\le$ is a two-place relation symbol of sort $\Q$, and $\W$ (the
worldview relation) is a $d+2$-place relation symbol the first two
arguments of which are of sort $\B$ and the rest are of sort $\Q$.

Relations $\Ob(o)$, $\IOb(m)$ and $\Ph(p)$ are translated as
``\textit{$o$ is an observer},'' ``\textit{$m$ is an inertial
  observer},'' and ``\textit{$p$ is a light signal},''
respectively. To speak about coordinatization of observers, we
translate relation $\W(k,b,x_1,x_2,\ldots,x_d)$ as ``\textit{body $k$
  coordinatizes body $b$ at space-time location $\langle x_1,
  x_2,\ldots,x_d\rangle$},'' (i.e., at space location $\langle
x_2,\ldots,x_d\rangle$ and instant $x_1$).

{\bf Quantity  terms} are the variables of sort $\Q$ and what can be
built from them by using the two-place operations $+$ and $\cdot$,
{\bf body terms} are only the variables of sort $\B$.
$\IOb(m)$, $\Ph(p,b)$, $\W(m,b,x_1,\ldots,x_d)$, $x=y$, and $x\le y$
where $m$, $p$, $b$, $x$, $y$, $x_1$, \ldots, $x_d$ are arbitrary
terms of the respective sorts are so-called {\bf atomic formulas} of
our first-order logic language. The {\bf formulas} are built up from
these atomic formulas by using the logical connectives \textit{not}
($\lnot$), \textit{and} ($\land$), \textit{or} ($\lor$),
\textit{implies} ($\rightarrow$), \textit{if-and-only-if}
($\leftrightarrow$) and the quantifiers \textit{exists} ($\exists$)
and \textit{for all} ($\forall$).

To make them easier to read, we omit the outermost universal
quantifiers from the formalizations of our axioms, i.e., all the free
variables are universally quantified.

We use the notation $\Q^n$ for the set of all $n$-tuples of elements
of $\Q$. If $\vx\in \Q^n$, we assume that $\vx=\langle
x_1,\ldots,x_n\rangle$, i.e., $x_i$ denotes the
$i$-th component of the $n$-tuple $\vx$. Specially, we write $\W(m,b,\vx)$ in
place of $\W(m,b,x_1,\dots,x_d)$, and we write $\forall \vx$ in place
of $\forall x_1\dots\forall x_d$, etc.

We use first-order logic set theory as a meta theory to speak about model
theoretical terms, such as models, validity, etc.  The {\bf models} of
this language are of the form
\begin{equation*}
{\mathfrak{M}} = \langle \B, \Q;\Ob_\mathfrak{M}, \IOb_\mathfrak{M},\Ph_\mathfrak{M},+_\mathfrak{M},\cdot_\mathfrak{M},\le_\mathfrak{M},\W_\mathfrak{M}\rangle,
\end{equation*}
where $\B$ and $\Q$ are nonempty sets, $\Ob_\mathfrak{M}$,
$\IOb_\mathfrak{M}$ and $\Ph_\mathfrak{M}$ are subsets of $\B$,
$+_\mathfrak{M}$ and $\cdot_\mathfrak{M}$ are binary functions and
$\le_\mathfrak{M}$ is a binary relation on $\Q$, and $\W_\mathfrak{M}$
is a subset of $\B\times \B\times \Q^d$.  Formulas are interpreted in
$\mathfrak{M}$ in the usual way.  For the precise definition of the
syntax and semantics of first-order logic, see, e.g., \cite[\S
  1.3]{CK}, \cite[\S 2.1, \S 2.2]{End}.

\section{Numbers required by special relativity}
\label{ax-s}

In this section, we will investigate our main question within special
relativity. To do so, first we formulate axioms for special relativity
in the logic language of the previous section.

Since the language above contains the concept of quantities (and that
of addition, multiplication and ordering), we can formulate statements
about numbers directly.  In our first axiom, we state some basic
properties of addition, multiplication and ordering true for real
numbers.\footnote{Using axiom \ax{AxOFiled} instead of assuming that
  the structure of quantities is the field of real numbers not just
  makes our theory more flexible, but also makes it possible to
  investigate our main question.}

\begin{description}
\item[\underline{\ax{AxOField}}]
 The quantity part $\langle \Q,+,\cdot,\le \rangle$ is an ordered field, i.e.,
\begin{itemize}
\item  $\langle\Q,+,\cdot\rangle$ is a field in the sense of abstract
algebra; and
\item 
the relation $\le$ is a linear ordering on $\Q$ such that  
\begin{itemize}
\item[i)] $x \le y\then x + z \le y + z$ and 
\item[ii)] $0 \le x \land 0 \le y\then 0 \le xy$
holds.
\end{itemize}
\end{itemize}
\end{description}

\ax{AxOField} is a ``mathematical" axiom in spirit. However, it has
physical (even empirical) relevance. Its physical relevance is that we
can add and multiply the outcomes of our measurements and some basic
rules apply to these operations. Physicists use all properties of the
real numbers tacitly, without stating explicitly which property is
assumed and why. The two properties of real numbers which are the most
difficult to defend from empirical point of view are the Archimedean
property, see \cite{Rosinger08}, \cite[\S 3.1]{Rosinger09},
\cite{Rosinger11a},\cite{Rosinger11b}, and the supremum
property,\footnote{The supremum property (i.e., every nonempty and
  bounded subset of the numbers has a least upper bound) implies the
  Archimedean property. So if we want to get ourselves free from the
  Archimedean property, we have to leave this one, too.} see the
remark after the introduction of \ax{CONT} on page \pageref{p-cont}.

The rest of our axioms on special relativity will speak about the
worldviews of inertial observers.  To formulate them, we use the
following concepts.  The {\bf time difference} of coordinate points
$\vx,\vy\in\Q^d$ is defined as:
\begin{equation*}
\timed(\vx,\vy)\de x_1-y_1. 
\end{equation*}
To speak about the spatial distance of any two coordinate points, we
have to use squared distance since it is possible that the distance of
two points is not amongst the quantities. For example, the distance of
points $\langle 0,0\rangle$ and $\langle 1,1\rangle$ is $\sqrt{2}$. So
in the field of rational numbers, $\langle 0,0\rangle$ and $\langle
1,1\rangle$ do not have distance but they have squared distance.  Therefore, we
define the {\bf squared spatial distance} of $\vx,\vy\in\Q^d$ as:
\begin{equation*}
  \sqspace(\vx,\vy)\de (x_2-y_2)^2+\ldots+(x_d-y_d)^2.
\end{equation*}
We denote the {\bf origin} of $\Q^n$ by $\vo$, i.e., $\vo\de \langle
0,\ldots,0\rangle$.  

The next axiom is the key axiom of our axiom system for special
relativity, it has an immediate physical meaning. This axiom is the
outcome of the Michelson-Morley experiment. It has been continuously
tested ever since then.  Nowadays it is tested by GPS technology.

\begin{description}
\item[\underline{\ax{AxPh}}] For any inertial observer, the speed of
  light is the same everywhere and in every direction (and it is
  finite). Furthermore, it is possible to send out a light signal in
  any direction (existing according to the coordinate system)
  everywhere:
\begin{multline*}
\IOb(m)\rightarrow \exists c_m\Big[c_m>0\land \forall \vx\vy
\\ \big(\exists p \big[\Ph(p)\land \W(m,p,\vx)\land
    \W(m,p,\vy)\big] \leftrightarrow \sqspace(\vx,\vy)=
  c_m^2\cdot\timed(\vx,\vy)^2\big)\Big].
\end{multline*}
\end{description}

Let us note here that \ax{AxPh} does not require (by itself) that the
speed of light is the same for every inertial observer.  It requires
only that the speed of light according to a fixed inertial observer is
a positive quantity which does not depend on the direction or the
location.

By \ax{AxPh}, we can define the {\bf speed of light} according to
inertial observer $m$ as the following binary relation:
\begin{multline*}
\ls(m,v)\defiff v>0  \land \forall \vx\vy\big(
\exists p \big[\Ph(p)\land \W(m,p,\vx)\land \W(m,p,\vy)\big]\\
\then \sqspace(\vx,\vy)= v^2\cdot\timed(\vx,\vy)^2\big). 
\end{multline*}

By \ax{AxPh}, there is one and only one speed $v$ for every inertial
observer $m$ such that $\ls(m,v)$ holds. From now on, we will denote
this unique speed by $\ls_m$.

Our next axiom connects the worldviews of different inertial observers
by saying that all observers coordinatize the same ``external" reality
(the same set of events).  By the {\bf event} occurring for observer
$m$ at point $\vx$, we mean the set of bodies $m$
coordinatizes at $\vx$:
\begin{equation*}
\ev_m(\vx)\de\{ b : \W(m,b,\vx)\}.
\end{equation*}

\begin{description}
\item[\underline{\ax{AxEv}}]
All inertial observers coordinatize the same set of events:
\begin{equation*}
\IOb(m)\land\IOb(k)\rightarrow  \exists \vy\, \forall b
\big[\W(m,b,\vx)\leftrightarrow\W(k,b,\vy)\big].
\end{equation*}
\end{description}
From now on, we will use $\ev_m(\vx)=\ev_k(\vy)$ to abbreviate the
subformula $\forall b [\W(m,b,\vx)\leftrightarrow\W(k,b,\vy)]$ of
\ax{AxEv}. The next two axioms are only simplifying ones.

\begin{description}
\item[\underline{\ax{AxSelf}}]
Any inertial observer is stationary relative to himself:
\begin{equation*}
\IOb(m)\rightarrow \forall \vx\big[\W(m,m,\vx) \leftrightarrow x_2=\ldots=x_d=0\big].
\end{equation*}
\end{description}
Our last axiom on inertial observers is a symmetry axiom saying that
they use the same units of measurement.

\begin{description}
\item[\underline{\ax{AxSymD}}]
Any two inertial observers agree as to the spatial distance between
two events if these two events are simultaneous for both of them;
furthermore, the speed of light is 1 for all observers:
\begin{multline*}
\IOb(m)\land\IOb(k) \land x_1=y_1\land x'_1=y'_1\land
\ev_m(\vx)=\ev_k(\vx')\\ \land
\ev_m(\vy)=\ev_k(\vy')\rightarrow \sqspace(\vx,\vy)=\sqspace(\vx',\vy'),
\text{ and }\\
\IOb(m)\rightarrow\exists
p\big[\Ph(p)\land\W(m,p,0,\ldots,0)\land\W(m,p,1,1,0,\ldots,0)\big].
\end{multline*}
\end{description}

Let us introduce an axiom system for special relativity as the collection
of the axioms above, if $d\ge 3$:
\begin{equation*}\label{specrel}
\ax{SpecRel} \de \{ \ax{AxOField}, \ax{AxPh}, \ax{AxEv}, \ax{AxSelf},
\ax{AxSymD}\}.
\end{equation*}

In relativity theory, we are often interested in comparing the
worldviews of two different observers. To do so, we introduce the
worldview transformation between observers $m$ and $k$ (in symbols,
$\w_{mk}$) as the binary relation on $\Q^d$ connecting the coordinate
points where $m$ and $k$ coordinatize the same (nonempty) events:
\begin{equation*}
\w_{mk}(\vx,\vy)\defiff \ev_m(\bar
x)=\ev_k(\vy)\neq\emptyset.
\end{equation*}

Map $P:\Q^d\rightarrow\Q^d$ is called a Poincar\'e transformation iff
it is an affine bijection having the following property
\begin{equation*}
\timed(\vx,\vy)^2-\sqspace(\vx,\vy)=\timed(\vx',\vy')^2-\sqspace(\vx',\vy')
\end{equation*}
for all $\vx,\vy,\vx',\vy'\in\Q^d$ for which  $P(\vx)=\vx'$ and $P(\vy)=\vy'$.

Theorem~\ref{thm-poi} shows that our streamlined axiom system
\ax{SpecRel} perfectly captures the kinematics of special relativity
since it implies that the worldview transformations between inertial
observers are the same as in the standard non-axiomatic approaches.
\begin{thm}\label{thm-poi}
Let $d\ge3$. Assume \ax{SpecRel}. Then $\w_{mk}$ is a Poincar\'e
transformation if $m$ and $k$ are inertial
observers.\footnote{Actually, axioms \ax{AxOField}, \ax{AxPh},
  \ax{AxEv}, and \ax{AxSymD} are enough to prove this statement, see
  Theorem~\ref{thm-poi0}.}
\end{thm}
We postpone the proof of Theorem~\ref{thm-poi} to
Section~\ref{proof-poi}, where we will prove a slightly stronger
result, see Theorem~\ref{thm-poi0}. For a similar result over
Euclidean fields, see, e.g., \cite[Thms. 1.4 \& 1.2]{AMNSamples},
\cite[Thm. 11.10]{logst}, \cite[Thm.3.1.4]{SzPhd}.

The so-called {\bf worldline} of body $b$ according to observer $m$ is
defined as follows:
\begin{equation*}
\wl_m(b)\de\{ \vx: \W(m,b,\vx)\}.
\end{equation*}
\begin{cor}\label{cor-line}
Let $d\ge3$. Assume \ax{SpecRel}. The $\wl_m(k)$ is a straight line if
$m$ and $k$ are inertial observers.\footnote{Axioms \ax{AxOField},
  \ax{AxPh}, \ax{AxEv}, and \ax{AxSelf} are enough to prove this
  statement since, by Theorem~\ref{thm-decomp}, axioms \ax{AxOField},
  \ax{AxPh}, and \ax{AxEv} imply that the worldview transformations
  take lines to lines and $\w_m(k)$ is the $\w_{km}$ image of the
  time-axis by axiom \ax{AxSelf}.}
\end{cor}

Let $m$ and $k$ be inertial observers. The {\bf squared speed} of $k$
according to $m$ is defined as follows:
\begin{multline*}
\sqspeed(m,k,v) \defiff \\ \exists \vx\vy\big[
\vx\neq\vy\land\W(m,k,\vx)\land\W(m,k,\vy)\land \sqspace(\vx,\vy)
=v\cdot\timed(\vx,\vy)^2\big].
\end{multline*}
By Corollary~\ref{cor-line}, \ax{SpecRel} implies that, for each $m,k\in\IOb$,
there is one and only one $v$ such that $\sqspeed(m,k,v)$ holds. From now on
let us denote this unique $v$ by $\sqspeed_m(k)$.
\begin{rem}
Even if $\langle\Q,+,\cdot,\le\rangle$ is the ordered field of
rational numbers, it is possible that the squared speed of an observer
is $2$. For example, $\sqspeed_m(k)=2$ if $d=3$ and inertial observers
$k$ goes trough points $\langle 0,0,0\rangle,\langle
1,1,1\rangle\in\rac^3$ according to inertial observer $m$. However,
some quantity cannot be the squared speed in some fields. For
example, the squared speed cannot be $3$ if
$\langle\Q,+,\cdot,\le\rangle$ is the ordered field of rational
numbers and $d=3$. This is so, because the equation $x^2+y^2=3 z^2$
does not have a nonzero solution over the natural numbers (if $x$, $y$
and $z$ are solutions, then $x$, $y$, and $z$ are divisible by $3^n$ for
all natural numbers $n$; hence $x=y=z=0$). Consequently, it does not
have a nonzero solution over the field of rational numbers.
\end{rem}
Corollary~\ref{cor-slow} states  basically that relatively moving inertial
observers' clocks slow down by the Lorentz factor
$\gamma=(1-v^2/c^2)^{-1/2}$ where $v$ is the relative speed of the
  observers.

\begin{cor}\label{cor-slow}
Let $d\ge 3$. Assume \ax{SpecRel}. Let $m,k\in\IOb$ and let
$\vx,\vy,\vx',\vy'\in\Q^d$ such that $\vx,\vy\in\wl_k(k)$,
$\w_{km}(\vx)=\vx'$ and $\w_{km}(\vy)=\vy'$. Then
\begin{equation}\label{e-slow}
\timed(\vx',\vy')^2 =\frac{\timed(\vx,\vy)^2}{1-\sqspeed_m(k)}.
\end{equation}
\end{cor}
\begin{proof}
Formula \eqref{e-slow} is always defined since $\sqspeed_m(k)$ cannot
be $1$ by Theorem~\ref{thm-poi}.  The case $\vx=\vy$ is trivial since,
in this case, both $\timed(\vx,\vy)$ and $\timed(\vx',\vy')$ are
$0$. So let us assume that $\vx\neq\vy$.  Since $\vx,\vy\in\wl_k(k)$,
we have that $\sqspace(\vx,\vy)=0$ by \ax{AxSelf}. By Theorem
\ref{thm-poi}, $\w_{km}$ is a Poincar\'e transformation. Therefore,
\begin{equation*}
\timed(\vx,\vy)^2=\timed(\vx',\vy')^2-\sqspace(\vx',\vy').
\end{equation*}
Consequently,
\begin{equation*}
\timed(\vx,\vy)^2=\timed(\vx',\vy')^2\left(1-\frac{\sqspace(\vx',\vy')}{\timed(\vx',\vy')^2}\right).
\end{equation*}
Hence, by the definition of $\sqspeed_m(k)$, we get
\begin{equation*}
\timed(\vx,\vy)^2=\timed(\vx',\vy')^2\left(1-\sqspeed_m(k)\right).
\end{equation*}
since  $\w_{km}(\vx)\neq\w_{km}(\vy)$ and $\w_{km}(\vx),\w_{km}(\vy)\in\wl_m(k)$.
\end{proof}

Theorem~\ref{thm-poi} and its consequences show that \ax{SpecRel}
captures special relativity well over every ordered field. It is a
natural question to ask what happens with these theorems if we assume
less about the quantities. This is one side of the question ``what are
the numbers?'', which is a whole research direction:
\begin{que}[Research direction]\label{que-rd}
What remains from the theorems of  \ax{SpecRel}, if we replace ordered fields with other algebraic structures, e.g., with ordered rings?
\end{que}

Here we concentrate on the other side of our question; namely,
``how can some physical assumptions implicitly enrich the structure of
quantities?''.  To investigate this question, let us now introduce
notation $\num_{n}(\ax{Th})$ for the class of the quantity parts of the
models of theory \ax{Th} if $d=n$:
\label{num}
\begin{multline*}
\num_{n}(\ax{Th})=\{\text{The quantity parts }\\ \langle
\Q,+,\cdot,\leq\rangle \text{ of the models of \ax{Th} if }
d=n \}.
\end{multline*}
The same way we use the notation $\mathfrak{Q}\in\num_{n}(\ax{Th})$
for ordered field $\mathfrak{Q}$ as the model theoretic notation
$\mathfrak{Q}\in Mod(\ax{AxField})$.

\begin{description}\label{axthexp}
\item[\underline{\ax{AxThExp}}] Inertial observers can move along any
  straight line with any speed less than the speed of light:
\begin{multline*}
\exists h \enskip \IOb(h)\land
\big(\IOb(m)\land \sqspace(\vx,\vy)<\ls_m^2\cdot\timed(\vx,\vy)^2
  \\ \then \exists k \big[\IOb(k)\land \W(m,k,\vx)\land\W(m,k,\vy)\big]\big).
\end{multline*}
\end{description}

Theorem~\ref{thm-eof} below shows that axiom
\ax{AxThExp} implies that positive numbers have square roots if
\ax{SpecRel} is assumed.
\begin{thm}\label{thm-eof}If $n\ge3$, 
\begin{equation*} 
\num_{n}(\ax{SpecRel} + \ax{AxThExp})=\{\text{\,Euclidean fields\,}\}.
\end{equation*}
\end{thm}

\begin{proof}
By Theorem 3.8.7 of \cite{pezsgo}, we have that
\ax{SpecRel} + \ax{AxThExp} has a model over every Euclidean
field. Consequently,
\begin{equation*}
\num_{n}(\ax{SpecRel} + \ax{AxThExp})\supseteq\{\text{\,Euclidean
  fields\,}\}.
\end{equation*}

To show the converse inclusion, we have to prove that every positive
quantity has a square root in every model of \ax{SpecRel} +
\ax{AxThExp}.  To do so, let $x\in Q$ be a positive quantity. We have
to show that $x$ has a square root in $\Q$.

First we will prove that $1-v^2$ has a square root if $v\in\Q$ and
$0\le v< 1$. To do so, let $v\in\Q$ for which $0\le v<1$. Let
$\vy=\langle 1, v,0,\ldots,0\rangle$. By \ax{AxTheExp} there are
inertial observers $m$ and $k$ such that $\vo,\vy\in\wl_m(k)$.  By
Corollary~\ref{cor-line}, $\wl_m(k)$ is a line. Thus
$\sqspeed_m(k)=v^2$.  Therefore, there is a $z\in\Q$ such that
$1-v^2=z^2$ (i.e., $1-v^2$ has a square root in $\Q$) by \ax{AxField}
and Corollary~\ref{cor-slow}.

From \ax{AxField}, it is easy to show that 
\begin{equation*}
x=\left({\frac{x+1}{2}}\right)^2\cdot\left(1-\left({\frac{x-1}{x+1}}\right)^2\right)
\end{equation*}
for all $x\in\Q$. There is
a $z\in\Q$ such that
\begin{equation*}
1-\left(\frac{x-1}{x+1}\right)^2=z^2
\end{equation*}
since $0\le 1-\left({\frac{x-1}{x+1}}\right)^2<1$. So there is a
quantity, namely $\frac{x+1}{2}\cdot z$, which is the square root of
$x$; and that is what we wanted to prove.\end{proof}

\begin{rem}\label{rem-of}
Axiom \ax{AxThExp} cannot be omitted from Theorem~\ref{thm-eof} since
\ax{SpecRel} has a model over every ordered field, i.e.,
\begin{equation*}
\num_n(\ax{SpecRel})=\{\,\text{ordered fields} \,\}
\end{equation*}
for all $n\ge2$.  Moreover, it also has non trivial models in which
there are several observers moving relative to each other. We
conjecture that there is a model of \ax{SpecRel} such that the
possible speeds of observers are dense in interval $[0,1]$, see
Corollary~\ref{cor-arch} and Conjecture~\ref{conj-of} at pages
\pageref{cor-arch} and \pageref{conj-of}.
\end{rem}

In the proof of Theorem~\ref{thm-eof}, axiom \ax{AxSymD} is strongly
used since \ax{SpecRel} without \ax{AxSymD} does not imply the exact
ratio of the slowing down of moving clocks; \ax{SpecRel} without
\ax{AxSymD} only implies that at least one of two relatively moving
inertial observers' clocks run slow according to the other, see
\cite[\S 2.5]{pezsgo}.  So it is natural to investigate what remains of
Theorem~\ref{thm-eof} if we leave the symmetry axiom out. 
It is surprising but, in the case of $d=3$, Theorem~\ref{thm-eof}
remains valid even if we assume only \ax{\ls_m=1} from \ax{AxSymD}, see
Andr\'eka--Madar\'asz--N\'emeti \cite[Thm 3.6.17]{pezsgo}.  Now we
will show that even the assumption \ax{\ls_m=1} is not necessary.  To do so,
let us introduce the next axiom system
\begin{equation*}
\ax{SpecRel_0}=\ax{SpecRel}-\ax{AxSymD}.
\end{equation*}

\begin{thm}\label{thm-eof3}
\begin{equation*}
\num_{3}(\ax{SpecRel_0} + \ax{AxThExp})=\{\text{\,Euclidean
  fields\,}\}
\end{equation*}
\end{thm}
\begin{proof}
By Theorem \ref{thm-eof}, $\ax{SpecRel_0}$ + $\ax{AxThExp}$
has a model over every Euclidean field since even
$\ax{SpecRel}$ + $\ax{AxThExp}$ has one. So 
\begin{equation*}
\num_{3}(\ax{SpecRel_0} +
\ax{AxThExp})\supseteq \{\text{{Euclidean  fields\,}}\}.
\end{equation*}

To prove the converse inclusion, we have to prove that the quantity
structure of every model of \ax{SpecRel_0} + \ax{AxThExp} is a
Euclidean field if $d=3$.  By Theorem 3.6.17 of \cite{pezsgo}, the
quantity structures of the models of \ax{SpecRel_0} + \ax{AxThExp} +
\ax{\ls_m=1} are Euclidean fields if $d=3$. Therefore, it is enough to show that a model
of \ax{SpecRel_0} + \ax{AxThExp} + \ax{\ls_m=1} can be constructed from
every model of \ax{SpecRel_0} + \ax{AxThExp} without changing its
quantity structure.

Let $\M$ be an arbitrary $3$ dimensional model of \ax{SpecRel_0} +
\ax{AxThExp}. Let $\M^+$ be the model which is constructed from $\M$ by
rescaling the coordinatization of each inertial observer $m$ of $\M$
by the following map $\vx\mapsto\langle \ls_m x_1,x_2,\ldots
x_d\rangle$, i.e., rescaling the time of $m$ by the factor $\ls_m$. It
is clear that the speed of light becomes $1$ according to $m$ after the
rescaling. So \ax{\ls_m=1} holds in $\M^+$. It is also easy to see that
this rescaling does not change the validity of \ax{AxThExp} and the
other axioms of \ax{SpecRel_0}. Therefore, $\M^+$ is a model of axiom
system \ax{SpecRel_0} + \ax{AxThExp} + \ax{\ls_m=1}. By the
construction, the quantity parts of $\M^+$ and $\M$ are the
same. Consequently, the quantity part of $\M$ is a Euclidean
field. This completes our proof since $\M$ was an arbitrary model of
axiom system \ax{SpecRel_0} + \ax{AxThExp}.
\end{proof}

Until recently, it was unsolved whether Theorem \ref{thm-eof3} is
valid or not in any higher dimension (see \cite[Questions 3.6.17 and
  3.6.19]{pezsgo}) when Hajnal Andr\'eka has provided counterexamples
in the even dimensions, i.e., the following is true:
\begin{thm}\label{thm-eof4}
\begin{equation*}
\num_{2k}(\ax{SpecRel_0} +
  \ax{AxThExp} + \ls_m=1)\supsetneqq\{\text{{Euclidean  fields\,}}\}
\end{equation*}
\end{thm}
For the proof of Theorem~\ref{thm-eof4}, see \cite{SqrtinSR}.

The existence of models of \ax{SpecRel_0} + \ax{AxThExp} over non
Euclidean fields is a surprising result since it is natural to
conjecture that a 3 dimensional model can be constructed from any
$d\ge 4$ dimensional model of \ax{SpecRel_0} + \ax{AxThExp} without
changing its quantity structure (by ``cutting out'' a 3 dimensional
part). Clearly, such a construction would imply Theorem~\ref{thm-eof3}
in any dimension higher than $3$, too.  It is interesting to note that
this kind of construction works if the quantity structure is a
Euclidean field.

Theorem \ref{thm-eof4} only shows that there are models of
\ax{SpecRel_0} + \ax{AxThExp} over some non-Euclidean
fields. However, the question ``what are the fields over which
\ax{SpecRel_0} + \ax{AxThExp} has a model?'' is still unsolved even in 4
dimension:
\begin{que}\label{que-thx}
Exactly which ordered fields are the elements of the class
$\num_{n}(\ax{SpecRel_0} + \ax{AxThExp})$ if $n\ge4$.
\end{que}

Without adding extra axioms to \ax{SpecRel} + \ax{AxThExp}, it does
not imply that the structure of numbers has to be a Euclidean field if
$d=2$. One of the reasons for this fact is that, if $d=2$, the axioms
of \ax{SpecRel} do not imply that the world lines of inertial
observers are straight lines. So we have to add it as an extra axiom
stating this (\ax{AxLine}). For a precise formulation of \ax{AxLine},
see, e.g., \cite[p.620]{logst}. Another reason is that, if $d=2$,
there are no two events which are simultaneous according to two
relatively moving observers.  Therefore, \ax{AxSymD} states nothing if
$d=2$. So we have to change this axiom. For example, we may replace
\ax{AxSymD} with the statement ``moving observers see each others
clock the same way and \ax{\ls_m=1}'' (\ax{AxSymT}). For a precise
formulation of the first part of \ax{AxSymT}, see, e.g.,
\cite[p.8]{AMNSamples}, \cite[p.20]{SzPhd}.  Actually, \ax{AxSymT} is
equivalent to \ax{AxSymD} if \ax{SpecRel_0} + \ax{\ls_m=1} is assumed
and $d\ge3$, see, e.g., \cite[Thm.3.1.4]{SzPhd}.
\begin{que}\label{que-2}
Does \ax{SpecRel} + \ax{AxThExp} + \ax{AxLine} + \ax{AxSymT} imply
that the quantities form a Euclidean field if $d=2$? If not, what
further natural axioms we have to assume to prove that the quantities
form a Euclidean field?
\end{que}

Since our measurements have only finite accuracy, it is natural to
assume \ax{AxThExp} only approximately. To introduce an approximated
version of \ax{AxThExp}, we need some definitions.  The {\bf space
  component} of coordinate point $\vx\in\Q^d$ is defined as $ \vx_s\de
\langle x_2,\ldots,x_d\rangle$.  The {\bf squared Euclidean distance}
of $\vx,\vy\in\Q^d$ is defined as
\begin{equation*}
\sqdist(\vx,\vy)\de (x_1-y_1)^2+\ldots+(x_d-y_d)^2
\end{equation*}
and the {\bf difference of} $\vx,\vy\in\Q^d$ is defined as
\begin{equation*}
  \vx-\vy\de\langle x_1-y_1,\ldots,x_d-y_d\rangle.
\end{equation*}

Let the {\bf squared Euclidean length} of $\vx\in\Q^d$ be defined as
\begin{equation*}
\sqlength(\vx)\de{x_1^2+\ldots+x_d^2}.
\end{equation*}

\begin{description}
\item[\underline{\ax{AxThExp^-}}] Inertial observers can move roughly
  with any speed less than the speed of light roughly in any
  direction:
\begin{multline*}  
\exists h \enskip \IOb(h) \land \Big(\IOb(m)\land \varepsilon >0 \land
\sqlength(\vv_s)<\ls_m^2 \\\land v_1=1 \rightarrow \exists \vw \Big[
\sqdist(\vw,\vv)<\varepsilon \land \forall \vx\vy\, 
\exists \lambda\big( \vx-\vy=\lambda \vw \\ \rightarrow \exists
k\big[ \IOb(m)\land \W(m,k,\vy)\land\W(m,k,\vy)\big]\big)\Big]\Big).
\end{multline*}
\end{description}

A model of \ax{SpecRel} + \ax{AxThExp^-} can be constructed over the
field of rational numbers, i.e., the following is true:
\begin{thm}\label{thm-rac}
\begin{equation*}
\rac\in\num_{n}(\ax{SpecRel} + \ax{AxThExp^-})
\end{equation*}
\end{thm}
For the proof of Theorem~\ref{thm-rac}, see \cite{MSzRac}.

An ordered field is called {\bf Archimedean ordered field} iff for all
$a$, there is a natural number $n$ such that
\begin{equation*}
a<\underbrace{1+\ldots+1}_n
\end{equation*}
holds.  By Pickert--Hion Theorem, every Archimedean ordered field is
isomorphic to subfield of the field of real numbers, see, e.g.,
\cite[\S VIII]{fuchs}, \cite[C.44.2]{CHA}. Consequently, the field of
rational numbers is dense in any Archimedean ordered field since it is
dense in the field of real numbers. Therefore, the following is a
corollary of Theorem \ref{thm-rac}.
\begin{cor}\label{cor-arch}
\begin{equation*}
\{\text{Archimedean ordered
  fields\,}\}\!\subsetneqq\!\num_{n}(\ax{SpecRel}\! +\!
\ax{AxThExp^-})\end{equation*}
\end{cor}

By L\"ovenheim--Skolem Theorem it is clear that
$\num_{n}(\ax{SpecRel}\! +\! \ax{AxThExp^-})$ cannot be the class of
Archimedean ordered fields since it has elements of arbitrarily large
cardinality while Archimedean ordered fields are subsets of the field
of real numbers by Pickert--Hion Theorem.  The question ``exactly
which ordered fields can be the quantity structures of theory
\ax{SpecRel} + \ax{AxThExp^-}?'' is open. We conjecture that there is
a model of \ax{SpecRel} + \ax{AxThExp^-} over every ordered field,
i.e.:
\begin{conj}\label{conj-of}
\begin{equation*}
\num_{n}(\ax{SpecRel} + \ax{AxThExp^-})=\{\,\text{ordered fields\,}\}
\end{equation*}
\end{conj}

\section{Numbers implied by accelerated observers}

Now we are going to investigate what happens with the possible
structures of quantities if we extend our theory \ax{SpecRel} with
accelerated observers. To do so, let us recall our first-order logic
axiom system of accelerated observers \ax{AccRel}.  The key axiom of
\ax{AccRel} is the following:
\begin{description}
\item[\underline{\ax{AxCmv}}]
 At each moment of its worldline, each observer
sees the nearby world for a short while as an inertial
   observer does.  
\end{description}
For formalization of \ax{AxCmv}, see \cite{SzPhd}. In \ax{AccRel} we
will also use the following localized version of axioms \ax{AxEv} and
\ax{AxSelf} of \ax{SpecRel}.

\begin{description}\label{axev-}
\item[\underline{\ax{AxEv^-}}] Observers coordinatize all the events in which they participate:
\begin{equation*}
\Ob(k)\land  \W(m,k,\vx)\rightarrow\exists
\vy\enskip \ev_m(\vx)=\ev_k(\vy).
\end{equation*}
\end{description}

\begin{description}
\item[\underline{\ax{AxSelf^-}}] In his own worldview, the worldline
  of any observer is an interval of the time-axis containing all the
  coordinate points of the time-axis where the observer sees
  something:
\begin{multline*}
 \big[\W(m,m,\vx)\then x_2=\ldots=x_d=0\big]
 \land\\ \big[\W(m,m,\vy)\land\W(m,m,\vz)\land x_1<t<y_1\then
   \W(m,m,t,0,\ldots,0)\big] \land\\ \exists b \big[\enskip
   \W(m,b,t,0,\ldots,0) \then \W(m,m,t,0,\ldots,0)\big].
\end{multline*}
\end{description}

Let us now introduce a promising theory of accelerated observers as
\ax{SpecRel} extended with the three axioms above.
\begin{equation*}
\ax{AccRel_0}= \ax{SpecRel}\cup
  \{\ax{AxCmv},\ax{AxEv^-},\ax{AxSf^-}\}
\end{equation*}

Since \ax{AxCmv} ties the behavior of accelerated observers to the
inertial ones and \ax{SpecRel} captures the kinematics of special
relativity perfectly by Theorem~\ref{thm-poi}, it is quite natural to
think that \ax{AccRel_0} is a strong enough theory of accelerated
observers to prove the most fundamental results about accelerated
observers.  However, \ax{AccRel_0} does not even imply the most basic
predictions about accelerated observers such as the twin paradox or
that stationary observers measure the same time between two events
\cite{twp}, \cite[\S 7]{SzPhd}. Moreover, it can be proved that even if
we add the whole firs-order logic theory of real numbers to
\ax{AccRel_0} is not enough to get a theory that implies the twin
paradox, see, e.g., \cite{twp}, \cite[\S 7]{SzPhd}.

In the models of \ax{AccRel_0} in which \ax{TwP} is not true there are
some definable gaps in the number line. Our next assumption is an
axiom scheme excluding these gaps.
\begin{description}
\label{p-cont}
\item[\underline{\ax{CONT}}] Every  parametrically definable,
  bounded and nonempty subset of $\Q$ has a supremum (i.e., least upper bound) with respect to $\le$.
\end{description}
\noindent In \ax{CONT} ``definable'' means ``definable in the language
of \ax{AccRel}, parametrically.'' For a precise formulation of \ax{CONT},
see \cite[p.692]{twp} or \cite[\S 10.1]{SzPhd}.

That \ax{CONT} requires the existence of supremum only for sets
definable in the language of \ax{AccRel} instead of every set is
important because it makes this postulate closer to the
physical/empirical level. This is true because \ax{CONT} does not
speak about ``any fancy subset'' of the quantities, but just about
those ``physically meaningful'' sets which can be defined in the
language of our (physical) theory.

Our axiom scheme of continuity (\ax{CONT}) is a ``mathematical
axiom" in spirit. It is Tarski's first-order logic version of
Hilbert's continuity axiom in his axiomatization of geometry, see
\cite[pp.161-162]{Gol}, fitted to the language of \ax{AccRel}.

When $\Q$ is the ordered field of real numbers, \ax{CONT} is automatically
true.
Let us introduce our axioms system \ax{AccRel} as the extension of \ax{AccRel_0} by axiom scheme \ax{CONT}. 
\begin{equation*}
\ax{AccRel}=\ax{AccRel_0} + \ax{CONT}
\end{equation*}

It can be proved that axiom system \ax{AccRel} implies the twin
paradox, see \cite{twp}, \cite[\S 7.2]{SzPhd}.

An ordered field is called {\bf real closed field} if a first-order
logic sentence of the language of ordered fields is true in it exactly
when it is true in the field of real numbers, or equivalently if it is
Euclidean and every polynomial of odd degree has a root in it, see,
e.g., \cite{tarski-dmethod}.

\begin{thm}\label{thm-rc}
\begin{equation*}
\num_{n}(\ax{AccRel})=\{\text{\,real closed fields\,}\}
\end{equation*}
\end{thm}

\begin{proof}
There is a model of \ax{AccRel} over every real closed field since
every model of \ax{SpecRel} over a real closed field in which
$\B=\Ph\cup\IOb$ is a model of \ax{AccRel} and \ax{SpecRel} has a
model even over every Euclidean ordered field by Theorem~\ref{thm-eof}.

Axiom schema \ax{CONT} is stronger than the whole first-order logic
theory of real numbers, see, e.g.,
\cite[Prop.~10.1.2]{SzPhd}. Consequently, if axiom \ax{AxOField} is assumed,
\ax{CONT} by itself implies that the quantities are real closed
fields.
\end{proof}

\section{Numbers implied by uniformly accelerated observers}

We have seen that assuming existence of observers can ensure the
existence of numbers.  So let us investigate another axiom of this
kind.  

The next axiom ensures the existence of uniformly accelerated
observers.  To introduce it, let us define the {\bf life-curve}
$\lc_{m}(k)$ of observer $k$ according to observer $m$ as the
worldline of $k$ according to $m$ {\it parametrized by the time
  measured by $k$},\label{life-curve} formally:
\begin{equation*}
\lc_{m}(k)\de\{\, \langle t,\vx \rangle\in \Q\times \Q^d \::
\:\exists \vy\enskip k\in \ev_k(\vy)=\ev_m(\vx)\land y_1=t\,\}.
\end{equation*}

\begin{description}
\item[\underline{\ax{Ax\exists UnifOb}}] It is possible to accelerate
  an observer uniformly:\footnote{In relativity theory, uniformly
    accelerated observers are moving along hyperbolas, see, e.g.,
    \cite[\S 3.8, pp.37-38]{dinverno}, \cite[\S 6]{MTW}, \cite[\S
      12.4, pp.267-272]{Rindler}.}
\begin{multline*}
\IOb(m)\rightarrow \exists k \Big[ \Ob(k) \land \dom \lc_m(k)=\Q\\\land
  \forall \vx \big[ \vx\in \ran \lc_m(k) \leftrightarrow
    x_2^2-x_1^2=a^2\land x_3=\ldots=x_d=0\big]\Big].
\end{multline*}
\end{description}

\begin{thm}\label{thm-e}
Let $d\ge3$. Assume \ax{AccRel} and \ax{Ax\exists UnifOb}. Then there
is a definable differentiable function $E:\Q\rightarrow\Q$ such that
$\ran E=\Q^+=[0,\infty)$, $\frac{dE}{dt}=E$ and $E(-t)=1/E(t)$ for all
  $t\in\Q$.
\end{thm}
Let $\bar \rac\cap \R$ denote the ordered field of real algebraic
numbers.  Theorem~\ref{thm-e} implies that the ordered field of
algebraic real numbers cannot be the structure of quantities of theory
\ax{AccRel} + \ax{Ax\exists UnifOb}:
\begin{thm}\label{thm-noalg}
Let $n\ge3$.  
\begin{equation*}
\bar\rac\cap\R\not\in\num_n(\ax{AccRel} + \ax{Ax\exists UnifOb})
\end{equation*}
\end{thm}
See \cite{Accwnst} for proofs and more details of
Theorems~\ref{thm-e} and \ref{thm-noalg}.

\begin{rem}\label{rem-noalg}
By Theorem~\ref{thm-noalg}, if $n\ge3$, $\num_n(\ax{AccRel} +
\ax{Ax\exists UnifOb})$ is not an elementary class of ordered fields,
i.e., it is not a first-order logic axiomatizable class in the
language of ordered fields. Of course, it is a pseudoelementary class,
i.e., it is a reduct of an elementary class in a richer language.
\end{rem}

By Theorem~\ref{thm-noalg}, we know that not every real closed field
can be the quantity structure of \ax{AccRel} + \ax{Ax\exists
  UnifOb}. For example, the field of real algebraic numbers cannot be
the quantity structure of \ax{AccRel} + \ax{Ax\exists UnifOb}.
However, the problem that exactly which ordered fields can be the
quantity structures of \ax{AccRel} + \ax{Ax\exists UnifOb} is still
open:
\begin{que}\label{que-uob}
Exactly which ordered fields are the elements of classes
$\num_{n}(\ax{AccRel} + \ax{Ax\exists UnifOb})$ and
$\num_{n}(\ax{AccRel_0} + \ax{Ax\exists UnifOb})$?
\end{que}

Theorem~\ref{thm-e} suggests that the answer to Question~\ref{que-uob}
may have something to do with ordered exponential fields, see, e.g.,
\cite[\S 4]{DW}, \cite{SK}.

\section{Numbers required by general relativity}
\label{sec-gr}

Let us now see some similar questions about the properties of numbers
implied by axioms of general relativity. To do so, let us recall our
axiom system \ax{GenRel} of general relativity formulated in the same
streamlined language as \ax{AccRel} and \ax{SpecRel}. \ax{GenRel}
contains the localized versions of the axioms of \ax{SpecRel} and the
postulate that the worldview transformations between observers are
differentiable maps, which is the localized version of the theorem of
\ax{SpecRel} stating that the worldview transformations between
inertial observers are affine maps, see Theorem~\ref{thm-poi}.  We
have already introduced the localized versions of axioms \ax{AxEv} and
\ax{AxSelf}, see \ax{AxEv^-} and \ax{AxSelf^-} at page
\pageref{axev-}. Now let us state the localized versions of \ax{AxPh}
and \ax{AxSymD}.\footnote{For technical reasons, in \ax{GenRel} we use
  an equivalent version of \ax{AxSymD}, and we introduce that the
  speed of light is $1$ in \ax{AxPh} instead of in \ax{AxSym^-}.}

\begin{description}
\item[\underline{\ax{AxPh^-}}] The velocity of photons an observer ``meets'' is 1
when they meet, and it is possible to send out a photon in each
direction where the observer stands.
\end{description}

\begin{description}
\item[\underline{\ax{AxSym^-}}] Meeting observers see each other's
  clocks slow down the same way.
\end{description}

\begin{description}\label{axdiff}
\item[\underline{\ax{AxDiff}}] The worldview transformations between
   observers are functions having linear approximations at each
  point of their domain (i.e., they are differentiable maps).
\end{description}

For a precise formulation of axioms \ax{AxPh^-}, \ax{AxSym^-}, and
\ax{AxDiff}, as well as a ``derivation'' of the axioms of \ax{GenRel}
from that of \ax{SpecRel}, see, e.g., \cite{Synthese}, \cite[\S 9]{SzPhd}.
\begin{equation*}
\ax{GenRel}\de\{\ax{AxOFiled},\ax{AxPh^-},\ax{AxEv^-},\ax{AxSelf^-},\ax{AxSym^-},\ax{AxDiff}\}\cup\ax{CONT}
\end{equation*}

Axiom system \ax{GenRel} captures general relativity well since it is
complete with respect the standard models of general relativity, i.e.,
with respect to Lorentzian manifolds, see, e.g.,
\cite[Thm.4.1]{Synthese}, \cite[\S 9]{SzPhd}.

We call the worldline of observer $m$ \textit{timelike geodesic}, if
each of its points has a neighborhood within which this observer
``maximizes measured time" between any two encountered events, see
Figure~\ref{fig-geod} for illustration and \cite{Synthese} for a
formal definition of timelike geodesics in the language of
\ax{GenRel}.

According to the definition above, if there are only a few observers,
then it is not a big deal that the worldline of $m$ is a timelike
geodesic (it is easy to be maximal if there are only a few to be
compared to). To generate a real competition for the rank of having
a timelike geodesic worldline, we postulate the existence of great many
observers by the following axiom scheme of comprehension.

\begin{description}
\item[\underline{\ax{COMPR}}] For any parametrically definable
  timelike curve in any observer's worldview, there is another observer
  whose worldline is the range of this curve.
\end{description}
A precise formulation of \ax{COMPR} can be obtained from that of its
analogue in \cite[p.679]{logst}. Let us now show that \ax{COMPR}
implies axiom \ax{Ax\exists UnifOb}, hence it requires at least as
much properties of numbers.

\begin{figure}
\psfrag{k}[rt][rt]{${m}$}
\psfrag{h}[lt][lt]{$\forall {k}$}
\psfrag{d}[l][l]{$\exists {\delta > 0}$}
\psfrag{x}[rb][rb]{$\forall \vx$}
\includegraphics[keepaspectratio, width=0.5\textwidth]{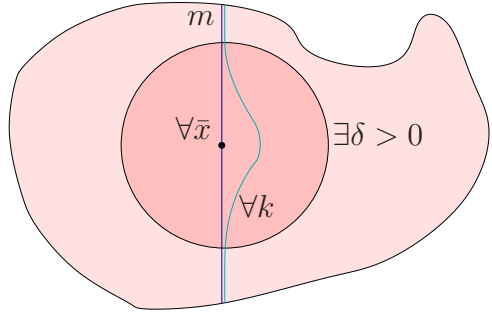}
\caption{Illustration for the definition of timelike geodesic in \ax{GenRel}}
\label{fig-geod}
\end{figure}

\begin{prop}\label{prop-uob}
\begin{equation*}
\num_{n}(\ax{AccRel} + \ax{COMPR})\subseteqq\num_{n}(\ax{AccRel} + \ax{Ax\exists UnifOb})
\end{equation*}
\end{prop}

\begin{proof}[\colorbox{proofbgcolor}{\textcolor{axcolor}{On the proof}}]
For all $a\in\Q$, the hyperbola (line if $a=0$) 
\begin{equation}
\{\vx :
x_2^2-x_1^2=a^2,x_3=\ldots=x_d=0 \}
\end{equation}
can be parametrized by the
definable timelike curve
\begin{equation}
\{\langle x_1,\vx\rangle: x_2^2-x_1^2=a^2,x_3=\ldots=x_d=0\}.
\end{equation}
So by \ax{COMPR}, there is an observer whose worldline is this set. So
\ax{COMPR} implies \ax{Ax\exists UnifOb}. Therefore, every  model of
\ax{AccRel} + \ax{COMPR} is a model of of \ax{AccRel} +
\ax{Ax\exists UnifOb}. Hence the possible quantity structures of
\ax{AccRel} + \ax{COMPR} is a subset of the possible quantity
structures of \ax{AccRel} + \ax{Ax\exists UnifOb}.
\end{proof}

It is also quite easy to show that \ax{GenRel} does not require more
properties of numbers than \ax{AccRel}.

\begin{prop}
\begin{equation*}
\num_{n}(\ax{AccRel} + \ax{AxDiff})\subseteqq\num_{n}(\ax{GenRel})
\end{equation*}
\begin{equation*}
\num_{n}(\ax{AccRel} + \ax{AxDiff} + \ax{COMPR})\subseteqq\num_{n}(\ax{GenRel} + \ax{COMPR})
\end{equation*}
\end{prop}

\begin{proof}[\colorbox{proofbgcolor}{\textcolor{proofcolor}{On the proof}}]
To prove this statement it is enough to show that the models of
\ax{AccRel} + \ax{AxDiff} are also models of \ax{GenRel}. Since
\ax{AxPh^-} and \ax{AxSym^-} are the only two axioms of \ax{GenRel}
which are not also contained in \ax{AccRel} + \ax{AxDiff}, we only
have to show that these two axioms are consequences of \ax{AccRel}.
Axioms \ax{AxPh^-} and \ax{AxSym^-} follow from \ax{AccRel} since they
are true for inertial observers in \ax{SpecRel} and by \ax{AxCmv}
accelerated observers locally see the world the same way as their
co-moving inertial observers.
\end{proof}

\begin{que}\label{que-gr}
Exactly which ordered fields are the elements of classes
$\num_{n}(\ax{AccRel} + \ax{COMPR})$
and $\num_{n}(\ax{GenRel} + \ax{COMPR})$?
\end{que}

Maybe the ordered field reducts of differentially closed fields of
Abraham Robinson, see \cite{MarkerDFC}, \cite{RobinsonDFC}, have to do
something with the answer to the question above.

\section{Proof of Theorem~\ref{thm-poi}}
\label{proof-poi}
In this section, we are going to prove Theorem~\ref{thm-poi}. To do
so, let us recall a version of Alexandrov--Zeeman theorem generalized
over fields. To state this theorem, we need some concepts.  Map
$\q:\Q^d\rightarrow\Q$ is a {\bf quadratic form} if 
\begin{equation}\label{eq-q}
\q(\lambda\vx)=\lambda^2\q(\vx)
\end{equation}
for all $\lambda\in\Q$ and $\vx\in\Q^d$, and 
\begin{equation}
(\vx,\vy)_\q\de \q(\vx+\vy)-\q(vx)-\q(\vy)
\end{equation}
 is a symmetric bilinear form.
 Quadratic form $\q$ is {\bf
  non-degenerate} if 
\begin{equation*}
\forall \vx\va\enskip (\vx,\va)_\q=0 \land \q(\va)=0\rightarrow
\va=\vo.
\end{equation*}
A map $f:\Q^d\rightarrow\Q^d$ is
called a {\bf semilinear map} iff there is a field automorphism $\alpha$
such that 
\begin{equation*}
f(\vx+\vy)=f(\vx)+f(\vy) \enskip\text{ and }\enskip f(\lambda
\vx)=\alpha(\lambda)f(\vx)
\end{equation*}
for all $\vx,\vy\in\Q^d$ and
$\lambda\in\Q$.  {\bf Witt index} of quadratic form $\q$ is the
maximal dimension of a subspace $X$ of $\Q^d$ with the property
$\q(\vx)=0$ for all $\vx\in X$.  {\bf $\q$-null cone} with vertex
$\va\in\Q^d$ is defined as 
\begin{equation*}
C(\va)=\{\vx: \q(\vx-\va)=0\}.
\end{equation*}
 
Now we are ready to recall the version of Alexandrov-Zeeman theorem we
need, see \cite{vro}, \cite{VKK}:
\begin{thm}[Vroegindewey]\label{thm-v}
Let $\langle\Q,+,\cdot\rangle$ be an commutative field. Let $d\ge 3$
and let $\q$ be a non-degenerate quadratic form
with Witt index 1. Then every bijection of $\Q^d$ taking $\q$-null
cones to $\q$-null cones is a composition of a translation and a
semilinear map $f$ with the property $\q\big(f(\vx)\big)=\lambda\alpha\big(\q(x)\big)$ for
some $\lambda\neq 0$ and field automorphism $\alpha$.
\end{thm}

We are going to apply Theorem~\ref{thm-v} to the worldview transformations of inertial observers in \ax{SpecRel}.
To do so, we need several definitions and lemmas. 

For all $c>0$, let us define the $c\,${\bf-Minkowski quadratic form} as
\begin{equation*}
\mu^2_c(\vx)=c\cdot x_1^2-x_2^2-\ldots-x_d^2.
\end{equation*}

\begin{lem}\label{lem-0}
Assume \ax{AxOField}. Let $\vx\in\Q^d$ be such that $x_1=0$ and $\mu^2_c(\vx)=0$. Then $\vx=\vo$.
\end{lem}
\begin{proof}
Since $x_1=0$ and $\mu^2_c(\vx)=0$, we have that $x_2^2+\ldots+
x_d^2=0$. This implies that $x_2=\ldots= x_d=0$ in ordered
fields. Hence $\vx=\vo$ as stated.
\end{proof}

\begin{rem}
Lemma~\ref{lem-0} is not valid in every field. For example, in the
field of complex numbers $\vx=\langle 0,1, i\rangle$ is a nonzero vector but $x_1=0$ and $\mu^2_1(\vx)=0$. 
\end{rem}

\begin{lem}\label{lem-witt}
Assume \ax{AxOField}. Let $c>0$. Then Minkowski quadratic form $\mu^2_c$ has Witt index 1.
\end{lem}
\begin{proof}
Let $\vx$ and $\vy$ be vectors such that
$\mu^2_c(\alpha\vx+\beta\vy)=0$ for all $\alpha,\beta\in\Q$. Let
$\vz=y_1\vx-x_1\vy$. Then $z_1=0$
and $\mu^2_c(\vz)=0$. Hence, by Lemma~\ref{lem-0}, $\vz=\vo$. So
$y_1\vx=x_1\vy$. Therefore, the subspace spanned by $\vx$ and $\vy$ is
$1$ dimensional. Thus the Witt index of $\mu^2_c$ is 1 as stated.
\end{proof}

The {\bf squared slope} of line $l$ is defined as
\begin{equation*}
\sqslope(l)=\frac{\sqspace(\vx,\vy)}{\timed(\vx,\vy)^2}
\end{equation*}
for all $\vx,\vy\in l$ for wihc $x_1\neq y_1$.

\begin{lem}\label{lem-noltr}
Assume \ax{AxOField}. Let $c>0$. There is no non-degenerate triangle
whose every side is of squared slope $c$.
\end{lem}
\begin{proof}
Let $\vx$, $\vy$, and $\vz$ be the vertices  of a triangle whose sides are of
squared slope $c$.  Then $c\cdot\timed(\vx,\vy)^2=\sqspace(\vx, \vy)$, $c\cdot\timed(\vy,\vz)^2=\sqspace(\vy,\vz)$, and
$c\cdot\timed(\vz,\vx)^2=\sqspace(\vz,\vx)$.
Let $\vp=\vy-\vx$ and $\vq=\vz-\vy$. Then
\begin{equation}\label{p}
c p_1^2=p_2^2+\ldots +p_d^2,
\end{equation}
\begin{equation}\label{q}
c q_1^2=q_2^2+\ldots +q_d^2,\text{ and}
\end{equation}
\begin{equation}\label{pq}
c(p_1+q_1)^2=(p_2+q_2)^2+\ldots +(p_d+q_d)^2.
\end{equation}
In other words $\mu^2_c(\vp)=\mu^2_c(\vq)=\mu^2_c(\vp+\vq)=0$. 
By subtracting equations \eqref{p} and \eqref{q} from equation \eqref{pq}, we get
\begin{equation}
2cp_1 q_1=2p_2q_2+\ldots +2p_dq_d.
\end{equation}
Let $\alpha$ and $\beta$ be arbitrary elements of $\Q$. Then 
\begin{multline}
\mu^2_c(\alpha\vp+\beta\vq)\\=\alpha^2\mu^2_c(\vp)+2\alpha\beta(c p_1 q_1-p_2q_2-\ldots -p_dq_d)+ \beta^2\mu^2_c(\vq)
\end{multline}
 for all $\alpha,\beta\in\Q$. Therefore,
 $\mu^2_c(\alpha\vp+\beta\vq)=0$.  By Lemma~\ref{lem-witt}, $\mu^2_c$
 has Witt index $1$. So $\vp$ and $\vq$ are in the same 1 dimensional
 subspace of $\Q^d$. Hence $\vx$, $\vy=\vx+\vp$, and $\vz=\vx+\vp+\vq$
 are collinear.
\end{proof}

The {\bf $f$-image of set $H$} is defined as follows:
\begin{equation*}
f[H]=\left\{\,b:\exists a\big[ a\in H \land f(a)=b\big]\,\right\}.
\end{equation*}

\begin{prop}\label{prop-bij}
Assume \ax{AxOField}, \ax{AxEv}, and \ax{AxPh}. Let $m,k\in\IOb$. Then
$\w_{mk}$ is a bijection of $\Q^d$ taking lines of squared slope $\ls^2_m$ to
lines of squared slope $\ls^2_k$.
\end{prop}
\begin{proof}
Let $m\in\IOb$ and let $\vx$ and $\vy$ be two distinct coordinate
points. Let $\vv\de\langle 1, \ls_m,0,\ldots, 0\rangle$ and $\vu\de\langle
1, -\ls_m,0,\ldots,0\rangle$.  By \ax{AxOField}, at most one of the
lines 
\begin{equation*}
l_1\de\{\vx+\lambda\cdot\vv :\lambda\in\Q\}\enskip\text{ and }\enskip l_2\de\{\vx+\lambda\cdot
\vu:\lambda\in\Q\}
\end{equation*}
can contain $\vy$ since $a\vv=b\vu$ implies
$a=b=0$. So, by \ax{AxPh}, there is a light signal which is in
$\ev_m(\vx)$ but not in $\ev_m(\vy)$ since $\sqslope(l_1)=\sqslope(l_2)=\ls^2_m$. Thus inertial observers see
different events at different coordinate points, i.e.,
$\ev_m(\vx)=\ev_m(\vy)$ implies $\vx=\vy$. Therefore, binary relation
$\w_{mk}$ is an injective function for all $m,k\in\IOb$.

Let $m,k\in\IOb$. By \ax{AxPh}, every inertial observer sees a
nonempty event in every coordinate point.  By \ax{AxEv}, inertial
observers coordinatize the same events. Therefore, for all $\vx\in\Q^d$, there
is a $\vy\in\Q^d$ such that $\w_{mk}(\vx)=\vy$. So $\dom \w_{mk}=\ran
\w_{km}=\Q^d$.  Consequently, $\w_{mk}$ is a bijection of $\Q^d$ for
all $m,k\in\IOb$.

Now we show that $\w_{mk}$ takes lines of squared slope $\ls^2_m$ to
lines of squared slope $\ls^2_k$.  To do so, let $l$ be a line of
squared slope $\ls^2_m$ and let $\vx$, $\vy$, $\vz$ be three distinct
points of $l$. By \ax{AxPh}, there are light signals $p_{xy}$,
$p_{yz}$, and $p_{zx}$ such that $p_{xy},p_{zx}\in\ev_m(\vx)$,
$p_{yz},p_{xy}\in\ev_m(\vy)$, and $p_{zx},p_{yz}\in\ev_m(\vz)$. Since
$\w_{mk}$ is a bijection, $\w_{mk}(\vx)$, $\w_{mk}(\vy)$, and
$\w_{mk}(\vz)$ are also distinct points. By the definition of
$\w_{mk}$, we have $p_{xy},p_{zx}\in\ev_k\big(\w_{mk}(\vx)\big)$,
$p_{yz},p_{xy}\in\ev_k\big(\w_{mk}(\vy)\big)$, and
$p_{zx},p_{yz}\in\ev_k\big(\w_{mk}(\vz)\big)$. So, by \ax{AxPh},
coordinate points $\w_{mk}(\vx)$, $\w_{mk}(\vy)$, and $\w_{mk}(\vz)$
form a triangle such that all of its sides are of squared slope
$\ls^2_k$. Therefore, by Lemma~\ref{lem-noltr}, they have to be on a
line of squared slope $\ls^2_k$. So the $\w_{mk}$-image of $l$ is a
subset of a line of $\ls^2_k$. Since $m$ and $k$ were arbitrary
inertial observers we also have that the $\w_{km}$-image of the line
containing $\w_{mk}[l]$ is the subset of a line of squared slope
$\ls^2_m$. Since $\w_{mk}$ is a bijection and its inverse is $\w_{km}$, we
have that $\w_{km}\big[\w_{mk}[l]\big]=l$. Consequently, $\w_{mk}[l]$
cannot be a proper subset of a line, but it has to be a whole line of squared
slope $\ls^2_k$.  This completes the proof of the proposition.
\end{proof}

\begin{cor}
Assume \ax{SpecRel}. Let $m$ and $k$ be inertial observers. Then
$\w_{mk}$ is a bijection of $\Q^d$ taking lines of squared slope 1 to
lines of squared slope 1. \qed
\end{cor}

Let us call a liner bijection of $\Q^d$ {\bf almost Lorentz
  transformation} iff there is a $\lambda\neq 0$ such that
$\mu^2_1\big(A(\vx)\big)=\lambda\mu^2_1(\vx)$ for all $\vx\in\Q^d$. 

We think of functions as special binary relations. Hence we compose
them as relations.  The {\bf composition} of binary relations $R$ and
$S$ is defined as:
\begin{equation*}\label{rcomp}
{R \fatsemi S}\de \{\langle a,c\rangle: \exists b\enskip R(a,b)\land  S(b,c) \}.
\end{equation*}
So $(g\fatsemi f)(x)=f\big(g(x)\big)$ if $f$ and $g$ are functions. We
will also use the notation $x\fatsemi g\fatsemi f$ for $(g\fatsemi
f)(x)$ because the latter is easier to grasp. In the same spirit, we will
sometimes use the notation $x\fatsemi f$ for $f(x)$. 
The {\bf inverse} of $R$ is defined as:
\begin{equation*}\label{rinv}
{R^{-1}}\de \{\langle a,b\rangle:  R(b,a) \}.
\end{equation*}

Let us introduce, for all $c>0$, the {\bf spatial distance and time
rescaling maps} as
\begin{equation*}
S_c(\vx)=\langle x_1,cx_2,\ldots, cx_d\rangle\enskip\text{ and }\enskip T_c(\vx)=\langle
cx_1, x_2,\ldots, x_d\rangle
\end{equation*}
for all $\vx\in\Q^d$. It is clear that
$T^{-1}_c=T_{1/c}$ and $S^{-1}_c=S_{1/c}$.

Let $\alpha$ be an automorphism of field $\langle
\Q,\cdot,+\rangle$ and let $\tilde\alpha$  be the map
$\tilde\alpha(\vx)=\langle \alpha(x_1),\ldots, \alpha(x_d) \rangle$ for all $\vx\in\Q^d$.
A map from $\Q^d$ to $\Q^d$ is called {\bf automorphism-induced-map}
if it is the form $\tilde\alpha$ for some automorphism
$\alpha$.\footnote{Let us note that we have not required that $\alpha$
  is order preserving.}

\begin{thm}\label{thm-decomp}
Let $d\ge3$.  Assume \ax{AxOField}, \ax{AxEv}, and \ax{AxPh}. Let
$m,k\in\IOb$. Then 
\begin{itemize}
\item $\w_{mk}=S^{-1}_{\ls_m}\fatsemi A\fatsemi
\tilde\alpha\fatsemi T\fatsemi S_{\ls_k}$ where $T$ is  a translation, $A$ is an
almost Lorentz transformation and  $\alpha$ is field automorphism.
\item $\w_{mk}=T_{\ls_m} \fatsemi A\fatsemi
\tilde\alpha\fatsemi T\fatsemi T^{-1}_{\ls_k}$ where $T$ is  a translation, $A$ is an
almost Lorentz transformation and  $\alpha$ is field automorphism.
\end{itemize}
\end{thm}

\begin{proof}
By definitions, $S_c$ and $T^{-1}_c$ are linear bijections of $\Q^d$
taking lines of squared slope $1$ to lines of squared slope
$c^2$. Therefore, by Proposition~\ref{prop-bij}, both maps
$S_{\ls_m}\fatsemi\w_{mk}\fatsemi S_{\ls_k}^{-1}$ and
$T_{\ls_m}^{-1}\fatsemi\w_{mk}\fatsemi T_{\ls_k}$ are bijections of
$\Q^d$ taking lines of squared slope 1 to lines of squared slope 1.
Since the $\mu^2_1$-null cone $C(\va)$ is the union of lines of
squared slope $1$ through $\va$, both
$S_{\ls_m}\fatsemi\w_{mk}\fatsemi S_{\ls_k}^{-1}$ and
$T_{\ls_m}^{-1}\fatsemi\w_{mk}\fatsemi T_{\ls_k}$ map $\mu^2_1$-null
cones to $\mu^2_1$-null cones. Therefore, by Theorem~\ref{thm-v} and
Lemma~\ref{lem-aldecomp}, they are compositions of an almost Lorentz
transformation $A$, a field-automorphism-induced map $\tilde\alpha$,
and a translation $T$.
\end{proof}

Some of the following statements assume only that the quantity part is
a field. Therefore, let us introduce the following axiom:
\begin{description}
\item[\underline{\ax{AxField}}]
 The quantity part $\langle \Q,+,\cdot\rangle$ is a (commutative) field.
\end{description}

\begin{lem}\label{lem-aut}
Assume \ax{AxField} and that $1+1\neq0$. Let $\alpha$ and $\beta$ be
two automorphisms of $\langle\Q,+,\cdot\rangle$ such that
$\alpha(a)^2=\beta(a)^2$ for all $a\in\Q$. Then $\alpha=\beta$.
\end{lem}

\begin{proof}
For all $a\in\Q$, we have that $\alpha(a)=\beta(a)$ or
$\alpha(a)=-\beta(a)$.  Let $a\in\Q$ such that $\alpha(a)=-\beta(a)$.
Then $\alpha(1+a)=1+\alpha(a)=1-\beta(a)$. Also
$\alpha(1+a)=\beta(1+a)=1+\beta(a)$ or
$\alpha(1+a)=-\beta(1+a)=-1-\beta(a)$. So $1-\beta(a)=1+\beta(a)$ or
$1-\beta(a)=-1-\beta(a)$. Therefore, $\beta(a)=0$ since $1+1\neq
0$. Hence $a=0$. Thus $\alpha(a)=\beta(a)$ for all $a\in\Q$.
\end{proof}

Let $\Id_H$ denote the {\bf identity map} from $H\subseteq \Q^d$ to
$H$, i.e., $\Id_H(\vx)=\vx$ for all $\vx\in H$.

\begin{rem}
It is easy to see that Lemma~\ref{lem-aut} is not valid if the field
has characteristic $2$, i.e., if $1+1=0$. For example, the $4$ element
field has two automorphisms $\Id$ and $\alpha$; and $\alpha^2=\Id^2$, but $\alpha\neq \Id$.
\end{rem}

\begin{lem}\label{lem-aldecomp}
Assume \ax{AxField}. Let $f:\Q^d\rightarrow\Q^d$ be a semilinear
transformation having the property
\begin{equation}\label{eq-al}
\mu^2_1\big(f(\vx)\big)=\lambda\alpha\big(\mu^2_1(\vx)\big)
\end{equation}
for some
$\lambda\neq0$ and field automorphism $\alpha$.  Then there are almost
Lorentz transformations $A$ and $A^*$ such that $f=\tilde\alpha\fatsemi
A =A^*\fatsemi \tilde\alpha$.
\end{lem}
\begin{proof}
Let $A$ be $\tilde\alpha^{-1}\fatsemi f$, i.e., 
\begin{equation}\label{eq-A}
A(\vx)=f\big(\tilde\alpha^{-1}(\vx)\big)
\end{equation}
for all $\vx\in\Q^d$.  
$A$ is a bijection since both $\tilde\alpha^{-1}$ and $f$ are so.
$A$ is additive, i.e.,
$A(\vx+\vy)=A(\vx)+A(\vy)$ for all $\vx,\vy\in\Q^d$, since
$\tilde\alpha^{-1}$ and $f$ are so.

Since $f$ is semilinear, there is a automorphism $\beta$ such
that 
\begin{equation}\label{eq-b}
f(a\vx)=\beta(a)f(\vx)
\end{equation}
for all $\vx\in\Q^d$ and $a\in\Q$. Consequently, we have
\begin{equation*}
  \mu^2_1\big(f(a\vx)\big)\stackrel{\eqref{eq-b}}{=}\mu^2_1\big(\beta(a)f(\vx)\big)\stackrel{\eqref{eq-q}}{=}\beta(a)^2\mu^2_1\big(f(\vx)\big)\stackrel{\eqref{eq-al}}{=}\beta(a)^2\lambda\alpha\big(\mu^2_1(\vx)\big)
\end{equation*}
and
\begin{equation*}
\mu^2_1\big(f(a\vx)\big)\stackrel{\eqref{eq-al}}{=}\lambda\alpha\big(\mu^2_1(a\vx)\big)\stackrel{\eqref{eq-q}}{=}\lambda\alpha\big(a^2\mu^2_1(\vx)\big)=\lambda\alpha(a)^2\alpha\big(\mu^2_1(\vx)\big).
\end{equation*}
for all $a\in\Q$.  Consequently,
$\lambda\beta(a)^2\alpha\big(\mu^2_1(\vx)\big)=\lambda\alpha(a)^2\alpha\big(\mu^2_1(\vx)\big)$
for all $a\in\Q$.  So $\alpha(a)^2=\beta(a)^2$ for all
$a\in\Q$. Therefore, by Lemma~\ref{lem-aut}, $\alpha=\beta$.
Consequently, equation \eqref{eq-b} becomes
\begin{equation}\label{eq-a}
f(a\vx)=\alpha(a)f(\vx).
\end{equation}
Thus $A$
is a linear bijection since
\begin{multline*}
A(a\vx)\stackrel{\eqref{eq-A}}{=}f\big(\tilde\alpha^{-1}(a \vx)\big)\stackrel{\eqref{eq-a}}{=}f\big(\alpha^{-1}(a)\tilde\alpha^{-1}(\vx)\big)\\=\alpha\big(\alpha^{-1}(a)\big)f\big(\tilde\alpha^{-1}(\vx)\big)=af\big(\tilde\alpha^{-1}(\vx)\big)\stackrel{\eqref{eq-A}}{=}a A(\vx)
\end{multline*}
for all $\vx\in\Q^d$ and $a\in\Q$.

Now we are going to show that $\mu^2_1\big(A(\vx)\big)=\lambda\mu^2_1(\vx)$ for all $\vx\in\Q^d$.
Let $\vx\in\Q^d$ and let $\vy=\tilde\alpha^{-1}(\vx)$. 
\begin{multline*}
\mu^2_1\big(A(\vx)\big)\stackrel{\eqref{eq-A}}{=}\mu^2_1\big(f\big(\tilde\alpha^{-1}(\vx)\big)\big)=\mu^2_1\big(f(\vy)\big)\\\stackrel{\eqref{eq-al}}{=}\lambda\alpha\big(\mu^2_1(\vy)\big)=\lambda\mu^2_1\big(\tilde\alpha(\vy)\big)=\lambda\mu^2_1(\vx).
\end{multline*}
This proves that $A$ is an almost Lorentz transformation; and $f=\tilde\alpha \fatsemi A$ by the
definition of $A$.

\medskip
We also have that $f=A^*\fatsemi \tilde\alpha$ for almost Lorentz transformation
$A^*=\tilde\alpha\fatsemi A\fatsemi \tilde \alpha^{-1}$.
\end{proof}

Vectors $\vx,\vy\in\Q^d$ are called {\bf orthogonal} in the Euclidean
sense, in symbols $\vx\perp_e\vy$, iff $x_1y_1+\ldots +x_dy_d=0$.

Vectors $\vx,\vy\in\Q^d$ are called {\bf Minkowski orthogonal}, in
symbols $\vx\perp_\mu\vy$, iff $(\vx,\vy)_{\mu^2_1}=0$, i.e.,
$x_1y_1=x_2y_2\ldots +x_dy_d$.

\begin{lem}\label{lem-alort}
Assume \ax{AxField}. Let $A$ be an almost
Lorentz transformation. Then $\vx\perp_\mu\vy$ iff $A(\vx)\perp_\mu
A(\vy)$ for all $\vx,\vy\in\Q^d$.
\end{lem}
\begin{proof}
By definition, $\vx\perp_\mu\vy$ iff $(\vx,\vy)_{\mu^2_1}=0$. Also by
definition
$\big(A(\vx),A(\vy)\big)_{\mu^2_1}=\mu^2_1\big(A(\vx)+A(\vy)\big)-\mu^2_1\big(A(\vx)\big)-\mu^2\big(A(\vy)\big)$. Since
$A$ is an almost Lorentz transformation, $\big(A(\vx),A(\vy)\big)_{\mu^2_1}=\lambda\cdot
(\vx,\vy)_{\mu^2_1}$ for some $\lambda\ne 0$.  Therefore, $\big(A(\vx),A(\vy)\big)_{\mu^2_1}=0$
iff $(\vx,\vy)_{\mu^2_1}=0$; and this is what we wanted to prove.
\end{proof}

\noindent
Let us introduce the {\bf time unit vector} as follows $\vet\de \langle 1, 0,\ldots,0\rangle$. 

\begin{prop}\label{prop-async}
Assume \ax{AxField}.  Let $A$ be an almost Lorentz
transformation. Then $y_1=0$ and $A(\vy)_1=0$ iff $A(\vet)\perp_e
A(\vy)$ and $A(\vy)_1=0$ for all $\vy\in\Q^d$.
\end{prop}

\begin{proof}
Let $\vy\in\Q^d$.  It is enough to show that $y_1=0$ is equivalent to
 $\vet\perp_e A(\vy)$ assuming that $A(\vy)_1=0$. It is clear that $y_1=0$
iff $\vet\perp_\mu\vy$. By Lemma~\ref{lem-alort}, $\vet\perp_\mu\vy$
iff $A(\vet)\perp_\mu A(\vy)$. Since $A(\vy)_1=0$, we have
$A(\vet)\perp_\mu A(\vy)$ iff $A(\vet)\perp_e A(\vy)$. Therefore,
$y_1=0$ iff $\vet\perp_e A(\vy)$ provided that
$A(\vy)_1=0$.
\end{proof}

Let $m$ and $k$ be inertial observers and let $\vx,\vy\in\Q^d$.
Events $\ev_m(\vx)$ and $\ev_m(\vy)$ are {\bf simultaneous} for $k$
iff $x'_1=y'_1$ for all $\vx'$ and $\vy'$ for which
$\ev_m(\vx)=\ev_k(\vx')$ and $\ev_m(\vy)=\ev_k(\vy')$. Events
$\ev_m(\vx)$ and $\ev_m(\vy)$ are {\bf separated orthogonally to the
  plane of motion} of $k$ according to $m$ iff $x_1=y_1$ and
$(\vx-\vy)\perp_e \big(\w_{km}(\vet)-\w_{km}(\vo)\big)$, see Figure~\ref{fig-ort}.

\begin{thm}\label{thm-async} 
Let $d\ge 3$. Assume \ax{AxOField}, \ax{AxPh}, and \ax{AxEv}. Let $m$
and $k$ be inertial observers and let $\vx,\vy\in\Q^d$. Events
$\ev_m(\vx)$ and $\ev_m(\vy)$ are simultaneous for both $m$ and $k$
iff $\ev_m(\vx)$ and $\ev_m(\vy)$ are separated orthogonally to the
plane of motion of $k$ according to $m$.\footnote{Specially, if $\sqspeed_m(k)=0$, the same events are simultaneous for $m$ and
  $k$.}
\end{thm}

\begin{proof}
Let $\vx'=\w_{mk}(\vx)$, $\vy'=\w_{mk}(\vy)$, and $\vv=\vy-\vx$, see
 Figure~\ref{fig-ort}.
\begin{figure}[h!btp]
\psfrag{m}[tl][tl]{${m}$}
\psfrag{k}[tl][tl]{${k}$}
\psfrag{o}[br][br]{$\vo$}
\psfrag{1}[br][br]{$1$}
\psfrag{x}[tr][tr]{$\vx$}
\psfrag{y}[tr][tr]{$\vy$}
\psfrag{xx}[tl][tl]{$\vx'=\w_{mk}(\vx)$}
\psfrag{yy}[br][br]{$\vy'=\w_{mk}(\vy)$}
\psfrag{wo}[tl][tl]{$\w_{km}(\vo)$}
\psfrag{w1}[tl][tl]{$\w_{km}(\vet)$}
\psfrag{v}[bl][bl]{$\vv$}
\psfrag{wmk}[bl][bl]{$\w_{mk}$}
\psfrag{wkm}[tl][rl]{$\w_{km}$}
\begin{center}
\includegraphics[keepaspectratio, width=\textwidth]{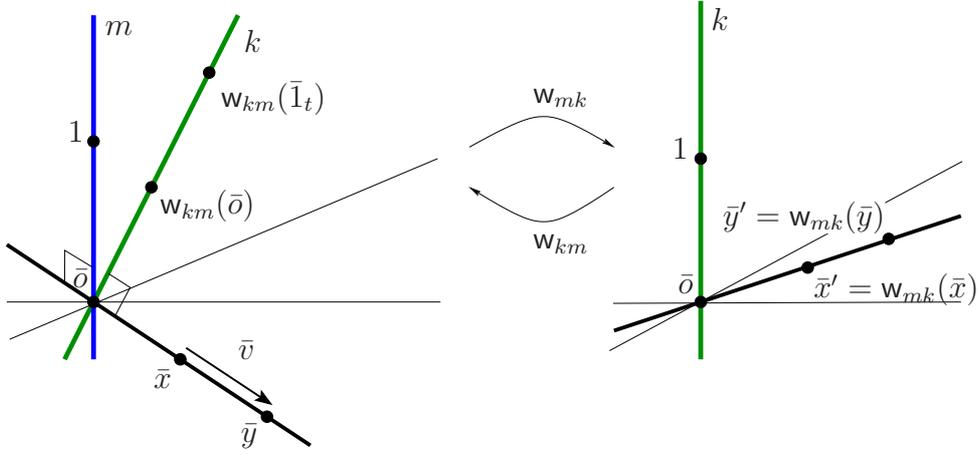}
\caption{Illustration for the proof of Theorem~\ref{thm-async}}
\label{fig-ort}
\end{center}
\end{figure}
By Theorem~\ref{thm-decomp}, $\w_{km}=S^{-1}_{\ls_k}\fatsemi A\fatsemi
\tilde\alpha \fatsemi T\fatsemi S_{\ls_m}$ for some field automorphism
$\alpha$, translation $T$ and almost Lorentz transformation $A$. Maps
$S_c$, $\tilde\alpha$ and $T$ do not change the facts whether
$\ev_m(\vx)$ and $\ev_m(\vy)$ are simultaneous for both $m$ and $k$;
and whether they are separated orthogonally to the plane of motion of
$k$ according to $m$. Therefore, we can assume, without loss of
generality, that $\w_{mk}$ is an almost Lorentz transformation.  Then
$\w_{km}(\vo)=\vo$. Therefore, events $\ev_m(\vx)$ and $\ev_m(\vy)$
are orthogonal to the plane of motion of $k$ according to $m$ iff
$v_1=0$ and $\vv\perp_e\w_{km}(\vet)$.  Let $\vv'=\w_{mk}(\vv)$, then
$\ev_m(\vx)$ and $\ev_m(\vy)$ are orthogonal to the plane of motion
iff $\w_{km}(\vv')_1=0$ and $\w_{km}(\vv')\perp_e\w_{km}(\vet)$.  By
Proposition~\ref{prop-async}, this is equivalent to
$\w_{km}(\vv')_1=0$ and $\vv'=0$. This means that $x_1=y_1$ and
$x'_1=y'_1$, i.e., that $\ev_m(\vx)$ and $\ev_m(\vy)$ are simultaneous
both for $m$ and $k$; and that is what we wanted to prove.
\end{proof}

Let $a\in\Q$ such that $a\neq 0$. Let us introduce {\bf dilation} $D_a$ as the
transformation mapping $\vx$ to $a\vx$ for all $\vx\in\Q^d$. It is clear that $D^{-1}_a=D_{1/a}$.

\begin{lem}\label{lem-decomp}
Assume \ax{AxField}. Let $A$ be an almost Lorentz transformation such
that $\mu^2_1\big(A(\vx)\big)=a^2\mu^2_1(\vx)$ for all
$\vx\in\Q^d$. There are a unique Lorentz transformation $L$ and a
unique dilation $D$ such that $A=D\fatsemi L=L \fatsemi
D$.
\end{lem}

\begin{proof}
Let $L$ be $D^{-1}_a\fatsemi A$. $L$ is a Lorentz transformation since 
\begin{equation*}
\mu^2_1\big(L(\vx)\big)=\mu^2_1\left(\frac{1}{a}A(\vx)\right)=\frac{1}{a^2}\mu^2_1\big(A(\vx)\big)=\frac{1}{a^2}a^2\mu^2_1(\vx)=\mu^2_1(\vx).
\end{equation*}
Therefore, $A=D_a\fatsemi L$ for Lorentz transformation $L$ and
dilation $D_a$.  Since $A$ is linear, $A=D^{-1}_a\fatsemi A \fatsemi
D_a$. Thus $A=D^{-1}_a\fatsemi D_a\fatsemi L \fatsemi D_a=L\fatsemi D_a$.

\medskip
If $A= D\fatsemi L$ for a Lorentz transformation $L$ and dilation $D$,
then $D$ has to be $D_a$ since $\mu^2_1\big(L(\vx)\big)=\mu^2_1(\vx)$
and $\mu^2_1\big(A(\vx)\big)=a^2\mu^2_1(\vx)$. Therefore, both $D$ and
$L$ are unique in the decomposition of $A$. The same proof works when
$A$ is decomposed as $A=L\fatsemi D$.
\end{proof}

\begin{lem}\label{lem-t1}
Assume \ax{AxOField}. Let $\vx,\vy\in\Q^d$ such that $\mu^2_1(\vx)>0$ and $(\vx,\vy)_{\mu^2_1}=0$. Then $\mu^2_1(\vy)<0$.
\end{lem}
\begin{proof}
Assume indirectly that $\mu^2_1(\vy)\ge0$, i.e., $y^2_1\ge
y^2_2+\ldots+y^2_d$. Since $x^2_1>x^2_2+\ldots+x^2_d$, we have that
$x^2_1y^2_1>(x_2^2+\ldots x_d^2)(y^2_2+\ldots+y^2_d)$. By
Cauchy--Schwarz inequality\footnote{For a simple proof of
  Cauchy--Schwarz inequality that works also in ordered fields, see
  \cite[\S 17]{proofs}.}\ we have $(x_2^2+\ldots x_d^2)(y^2_2+\ldots+y^2_d)\ge
  (x_2y_2+\ldots+ x_dy_d)^2$. Since $x_1y_1=x_2y_2+\ldots+ x_dy_d$, we
  have that $x^2_1y^2_1>(x_1y_1)^2$. This contradiction proves that
  $\mu^2_1(\vy)<0$.
\end{proof}

\begin{prop}\label{prop-al}
Let $d\ge3$. Assume \ax{AxOField}. Let $A$ be an almost Lorentz
transformation. Then there is a $\lambda>0$ such that
$\mu^2_1\big(A(\vx)\big)=\lambda\mu^2_1(\vx)$ for all $\vx\in\Q^d$.
\end{prop}

\begin{proof}
Since $A$ is an almost Lorentz transformation there is a $\lambda\neq
0$ such that $\mu^2_1\big(A(\vx)\big)=\lambda\mu^2_1(\vx)$ for all
$\vx\in\Q^d$. We are going to prove that this $\lambda$ has to be
positive.  Assume indirectly that $\lambda<0$. Let
$\vy=\langle0,1,0,\ldots,0\rangle$ and $\vz=\langle
0,0,1,0,\ldots,0\rangle$. Then $\mu^2_1(\vy)=\mu^2_1(\vz)=-1$ and
$(\vy,\vz)_{\mu^2_1}=0$. Let $\vy'=A(\vy)$ and $\vz'=A(\vz)$. Then
$\mu^2_1(\vy')>0$ and $\mu^2_1(\vz')>0$ since $\lambda<0$; and
$(\vy',\vz')_{\mu^2_1}=0$ by Lemma~\ref{lem-alort}.  These properties
of $\vy'$ and $\vz'$ contradict Lemma~\ref{lem-t1}. Therefore,
$\lambda>0$.
\end{proof}

\begin{rem} 
Proposition~\ref{prop-al} is not valid if $d=2$ since reflection
$\sigma_{tx}:\langle t,x\rangle\mapsto \langle x,t\rangle$ is an
almost Lorentz transformation and
$\mu^2_1\big(\sigma_{tx}(\vx)\big)=-\mu^2_1(\vx)$ for all $\vx\in\Q^2$.
\end{rem}


\begin{prop}\label{prop-aldecomp}
Let $d\ge3$. Assume that $\langle\Q, + ,\cdot\rangle$ is a Euclidean
 field. Then every almost Lorentz transformation is a
composition of a Lorentz transformation and a dilation.\qed
\end{prop}

\begin{proof}
The statement follows from Lemma~\ref{lem-decomp} and
Proposition~\ref{prop-al} since in Euclidean fields every
positive number has a square root.
\end{proof}

\begin{rem}
Proposition~\ref{prop-aldecomp} does not remain valid over arbitrary
ordered fields. To construct a counterexample, let $d=4$, $\langle \Q,
+, \cdot, \le \rangle$ be the ordered field of rational numbers, and
let $A$ be the following linear map $A(\vx)=\left\langle
\frac{3x_1+x_2}{2}, \frac{x_1+3x_2}{2}, x_3-x_4, x_3+x_4 \right\rangle$ for all
$\vx\in\rac^4$. It is straightforward to check that
$\mu^2_1\big(A(\vx)\big)=2\mu^2_1(\vx)$ for all $\vx\in\rac^4$; so $A$
is an almost Lorentz transformation. However, $A$ cannot be the
composition of a dilation $D$ and a Lorentz transformation $L$ over
the field of rational numbers since then $A$ would also be the
composition of $D$ and $L$ over the field of real numbers; and, by
Lemma~\ref{lem-decomp}, the dilation in the unique decomposition of
$A$ over the field of real numbers is $D_{\sqrt{2}}$, which does not
map $\rac^4$ to $\rac^4$.
\end{rem}

Now we are ready to prove Theorem~\ref{thm-poi}. In
Theorem~\ref{thm-poi0} we prove a slightly stronger result since we
will not use axiom \ax{AxSelf}.
\begin{thm}\label{thm-poi0}
Let $d\ge3$.  Assume \ax{AxOField}, \ax{AxEv}, \ax{AxPh}, and
\ax{AxSymD}. Let $m,k\in\IOb$. Then $\w_{mk}$ is a Poincar\'e
transformation.
\end{thm}

\begin{proof}
Since, by \ax{AxSymD}, the speed of light is $1$ according to every
inertial observer, $\w_{mk}$ is a composition of an almost Lorentz
transformation $A$, a field-automorphism-induced map $\tilde\alpha$
and a translation $T$ by Theorem \ref{thm-decomp}.  Specially,
$\w_{mk}$ maps lines to lines.

By \ax{AxOField}, there is a line $l$ orthogonal to the plane of
motion of $k$ according to $m$. By Theorem~\ref{thm-async}, both $l$
and $\w_{mk}[l]$ are horizontal. Therefore, by \ax{AxSymD}, $\w_{mk}$
maps $l$ to $\w_{mk}[l]$ preserving the squared Euclidean distances of
the points of $l$. Let $\vv$ be a direction vector of $l$.\footnote{That is,
  $\vv=\vy-\vx$ for two distinct points $\vx$ and $\vy$ of $l$.} Then,
by axiom \ax{AxSymD}, we have that
\begin{equation}\label{A}
\sqlength(x\vv)=\sqlength\big(\tilde\alpha (A(x\vv))\big)
\end{equation}
for all
$x\in\Q$ since both $x\vv$ and $\tilde\alpha\big( A( x\vv)\big)$ are
horizontal vectors. Since both $l$ and $\w_{mk}[l]$ are horizontal, we
have that 
\begin{equation}\label{eq-x}
\mu^2_1(\vv)=\sqlength(\vv)\enskip\text{ and }\enskip\mu^2_1\big(\tilde\alpha
(A (\vv))\big)=\sqlength\big(\tilde\alpha (A(\vv))\big).
\end{equation}
Since $A$ is an almost Lorentz transformation, there is a 
$\lambda\neq0$ such that 
\begin{equation}\label{eq-y}\mu^2_1\big(A
(\vx)\big)=\lambda\mu^2_1(\vx)
\end{equation}
for all $\vx\in\Q^d$. Thus
\begin{multline}\label{B}
\sqlength\big(\tilde\alpha (A(\vv))\big)\stackrel{\eqref{eq-x}}{=}\mu^2_1\big(\tilde\alpha (A (\vv))\big)=\alpha
\big(\mu^2_1(A(\vv))\big)\stackrel{\eqref{eq-y}}{=}\alpha\big(\lambda\mu^2_1(\vv)\big)\\=\alpha(\lambda)\alpha\big(\mu^2_1(\vv)\big)\stackrel{\eqref{eq-x}}{=}\alpha(\lambda)\alpha\big(\sqlength(\vv)\big)
\end{multline}
Therefore, by the fact that that $\sqlength(a\vy)=a^2\sqlength(\vy)$ and  Equations~\eqref{A} and \eqref{B}, we get
\begin{equation}\label{*}
x^2\sqlength(\vv)=\alpha(\lambda)\alpha(x)^2\alpha\big(\sqlength(\vv)\big)
\end{equation}
for all $x\in\Q$. Specially, 
\begin{equation}\label{**}
\sqlength(\vv)=\alpha(\lambda)\alpha\big(\sqlength(\vv)\big)
\end{equation}
by choosing $x=1$ in equation~\eqref{*}.  Equations \eqref{*} and
\eqref{**} imply that $x^2=\alpha(x)^2$ for all $x\in\Q$.
Consequently, $\alpha=\Id_\Q$ by Lemma~\ref{lem-aut}. Thus
$\tilde\alpha=\Id_{\Q^d}$ and $1=\alpha(\lambda)$ by equation
\eqref{*}. So $\lambda=1$, i.e., $A$ is a Lorentz transformation.

So $\tilde\alpha$ has to be the identity map and $A$
has to be a Lorentz transformation. Thus $\w_{mk}$ is a composition of
a Lorentz transformation and a translation, i.e., it is a Poincar\'e
transformation as it was stated.
\end{proof}

\section{Concluding remarks}
We have seen that the possible structures of quantities strongly
depend on the other axioms of spacetime. Typically, axioms requiring
the existence of additional observers reduce the possible structures
of quantities, see Theorems~\ref{thm-eof}, \ref{thm-eof3}, \ref{thm-e}
and Proposition~\ref{prop-uob}. We have proved several propositions
about the connection between spacetime axioms and the possible structures
of numbers. However, there are still great many open questions in this
research area, see Questions~\ref{que-rd}, \ref{que-thx}, \ref{que-2},
\ref{que-uob}, \ref{que-gr} at pages \pageref{que-rd},
\pageref{que-thx}, \pageref{que-2}, \pageref{que-uob},
\pageref{que-gr}, and Conjecture~\ref{conj-of} at page
\pageref{conj-of}.

\section{Acknowledgments}
This research is supported by the Hungarian Scientific Research Fund
for basic research grants No.~T81188 and No.~PD84093, as well as by a
Bolyai grant for J.~X.~Madar\'asz.

\bibliography{LogRelBib}
\bibliographystyle{plain}

\end{document}